\documentclass[11pt]{article}
\usepackage{paralist}
\usepackage{color}
\usepackage{bm}

\usepackage{soul}
\usepackage{mathtools}

\usepackage[cm]{fullpage}
\usepackage{amsfonts}
\usepackage{amssymb}

\usepackage{amsmath,enumitem}
\usepackage[usenames,dvipsnames]{xcolor}
\usepackage{tikz,pgf,color}
\usepackage{constants}
\usepackage{amsthm}
\makeatletter
\newtheorem*{rep@theorem}{\rep@title}
\newcommand{\newreptheorem}[2]{%
	\newenvironment{rep#1}[1]{%
		\def\rep@title{#2 \ref{##1}}%
		\begin{rep@theorem}}%
		{\end{rep@theorem}}}
\makeatother

\newcount\shortyear\newcount\shorthour\newcount\shortminute
\shorthour=\time\divide\shorthour by 60\shortyear=\shorthour
\multiply\shortyear by 60\shortminute=\time\advance\shortminute by
-\shortyear \shortyear=\year\advance\shortyear by -1900

\def\zeit{\number\shorthour:\ifnum\shortminute<10 0\number\shortminute
	\else\number\shortminute\fi}

\allowdisplaybreaks

\usepackage{ifpdf}
\newcommand{\mydriver}{hypertex}
\ifpdf
\renewcommand{\mydriver}{pdftex}
\fi
\usepackage[breaklinks,\mydriver]{hyperref}

\usepackage[margin=1in]{geometry}

\theoremstyle{plain}
\newtheorem{theorem}{Theorem}[section]
\newtheorem{fact}[theorem]{Fact}
\newtheorem{lemma}[theorem]{Lemma}
\newtheorem{corollary}[theorem]{Corollary}

\newtheorem{claim}[theorem]{Claim}
\newtheorem{definition}[theorem]{Definition}

\usepackage{algorithm}
\usepackage{algpseudocode}
\usepackage[normalem]{ulem}

\algtext*{EndWhile}
\algtext*{EndIf}
\algtext*{EndFor}

\DeclarePairedDelimiterX{\set}[2]{\{}{\}}{#1\,\delimsize|\,\mathopen{}#2}
\DeclarePairedDelimiterX{\abs}[1]{\lvert}{\rvert}{#1}

\DeclareMathOperator{\mincut}{mincut}
\DeclareMathOperator{\capacity}{cap}

\title{Improved Guarantees for Vertex Sparsification in Planar Graphs\footnote{The research leading to these results has received funding from the
		European Research Council under the European Union's Seventh	Framework Programme (FP/2007-2013) / ERC Grant Agreement no. 340506.}}
\author{Gramoz Goranci\footnote{University of Vienna, Faculty of Computer Science, Vienna, Austria. E-mail: \texttt{gramoz.goranci@univie.ac.at}.}
	\and
	Monika Henzinger\footnote{University of Vienna, Faculty of Computer Science, Vienna, Austria. E-mail: \texttt{monika.henzinger@univie.ac.at}.}
	\and 
	Pan Peng\footnote{University of Vienna, Faculty of Computer Science, Vienna, Austria. E-mail: \texttt{pan.peng@univie.ac.at}.}
}

\begin{document}

\begin{titlepage}
	\maketitle
	\thispagestyle{empty}

\begin{abstract}
Graph Sparsification aims at compressing large graphs into smaller ones while preserving important characteristics of the input graph. In this work we study Vertex Sparsifiers, i.e., sparsifiers whose goal is to reduce the number of vertices. We focus on the following notions:

(1) Given a digraph $G=(V,E)$ and terminal vertices $K \subset V$ with $|K| = k$, a (vertex) reachability sparsifier of $G$ is a digraph $H=(V_H,E_H)$, $K \subset V_H$ that preserves all reachability information among terminal pairs. In this work we introduce the notion of reachability-preserving minors (RPMs) , i.e., we require $H$ to be a minor of $G$. We show any directed graph $G$ admits a RPM $H$ of size $O(k^3)$, and if $G$ is planar, then the size of $H$ improves to $O(k^{2} \log k)$. We complement our upper-bound by showing that there exists an infinite family of grids such that any RPM must have $\Omega(k^{2})$ vertices.

(2) Given a weighted undirected graph $G=(V,E)$ and terminal vertices $K$ with $|K|=k$, an exact (vertex) cut sparsifier of $G$ is a graph $H$ with $K \subset V_H$ that preserves the value of minimum-cuts separating any bipartition of $K$. We show that planar graphs with all the $k$ terminals lying on the same face admit exact cut sparsifiers of size $O(k^{2})$ that are also planar. Our result extends to flow and distance sparsifiers. It improves the previous best-known bound of $O(k^22^{2k})$ for cut and flow sparsifiers by an exponential factor, and matches an $\Omega(k^2)$ lower-bound for this class of graphs.

 \end{abstract}
\end{titlepage}
\section{Introduction}
Very large graphs or networks are ubiquitous nowadays, from social networks to information networks. One natural and effective way of processing and analyzing such graphs is to compress or sparsify the graph into a smaller one that well preserves certain properties of the original graph. Such a sparsification can be obtained by reducing the number of \emph{edges}. Typical examples include cut sparsifiers~\cite{BenczurK96}, spectral sparsifiers~\cite{SpielmanT04}, spanners~\cite{ThorupZ05} and transitive reductions~\cite{AhoGU72}, which are subgraphs defined on the same vertex set of the original graph $G$ while having much smaller number of edges and still well preserving the cut structure, spectral properties, pairwise distances and transitive closure of $G$, respectively. Another way of performing sparsification is by reducing the number of \emph{vertices}, which is most appealing when only the properties among a subset of vertices (which are called \emph{terminals}) are of interest~(see e.g., \cite{moitra09,andoni,distancepreserving}). We call such small graphs \emph{vertex sparsifiers} of the original graph. In this paper, we will particularly focus on vertex reachability sparsifiers for \emph{directed} graphs and cut (and other related) sparsifiers for \emph{undirected} graphs. 

Vertex reachability sparsifiers in directed graphs is an important and fundamental notion in Graph Sparsification, which has been implicitly studied in the dynamic graph algorithms community~\cite{Subramanian93,DiksS07}, and explicitly in~\cite{katriel2005reachability}. Specifically, given a digraph $G=(V,E)$, $K \subset V$, a digraph $H=(V_H,E_H)$, $K \subset V_H$ is a (\emph{vertex}) \emph{reachability sparsifier} of $G$ if for any $x ,x' \in K$, there is a directed path from $x$ to $x'$ in $H$ iff there is a directed path from $x$ to $x'$ in $G$. If $|K|=k$, we call the digraph $G$ a \emph{$k$-terminal digraph}. Note that any $k$-terminal digraph $G$ always admits a trivial reachability vertex sparsifier $H$, which corresponds to the transitive closure restricted to the terminals. In this work, we initiate the study of \emph{reachability-preserving minors}, i.e., vertex reachability sparsifiers with $H$ required to be a minor of $G$. 
The restriction on $H$ being a minor of $G$ is desirable as it makes sure that $H$ is structurally similar to $G$, e.g., any minor of a planar graph remains planar. We ask the question whether general graphs admit reachability-preserving minors whose size can be bounded independently of the input graph $G$, and study it from both the lower- and upper-bound perspective.

For the notion of cut (and other related) sparsifiers, we are given a capacitated undirected graph $G=(V,E,c)$, and a set of terminals $K$ and our goal is to find a (capacitated undirected) graph $H=(V_H,E_H,c_H)$ with as few vertices as possible and $K \subseteq V_H$ such that the quantities like, cut value, multi-commodity flow and distance among terminal vertices in $H$ are the same as or close to the corresponding quantities in $G$. If $|K|=k$, we call the graph $G$ a \emph{$k$-terminal graph}. We say $H$ is a \emph{quality-$q$} (\emph{vertex}) \emph{cut sparsifier} of $G$, if for every bipartition $(U,K \setminus U)$ of the terminal set $K$, the value of the minimum cut separating $U$ from $K \setminus U$ in $G$ is within a factor of $q$ of the value of minimum cut separating $U$ from $K \setminus U$ in $H$. If $H$ is a quality-$1$ cut sparsifier, then it will be also called a \emph{mimicking network}~\cite{HagerupKNR98}. Similarly, we define flow and distance sparsifiers that (approximately) preserve multicommodity flows and distances among terminal pairs, respectively~(see Section~\ref{sec: UpperFlow} for formal definitions). These type of sparsifiers have proven useful in approximation algorithms~\cite{moitra09} and also find applications in network routing~\cite{chuzhoy2012routing}.

\subsection{Our Results}

\paragraph*{Reachability Sparsifiers.} Our first main contribution is the study of reachability-preserving minors. Although reachability is a weaker requirement in comparison to shortest path distances, directed graphs are usually much more cumbersome to deal with from the perspective of graph sparsification. Surprisingly, we show that general digraphs admit reachability-preserving minors with $O(k^{3})$ vertices, which is in contrast to the bound of $O(k^4)$ on the size of distance-preserving minors in undirected graphs by Krauthgamer et al.~\cite{distancepreserving}. 

\begin{theorem} \label{thmi: general_reachability} Given a $k$-terminal digraph $G$, there exists a reachability-preserving minor $H$ of $G$ with size $O(k^{3})$.
\end{theorem}

It might be interesting to compare the above result with the construction of \emph{reachability preserver} by Abbound and Bodwin~\cite{AB18reachability}, where the reachability preserver for a pair-set $P$ in a graph $G$ is defined to be a \emph{subgraph} of $G$ that preserves the reachability of all pairs in $P$. The size (i.e., the number of edges) of such preservers is shown to be at least $\Omega(n^{2/(d+1)}|P|^{(d-1)/d})$, for any integer $d\geq 2$, which is in sharp contrast to our upper bound $O(|P|^{3/2})$ on the size of reachability-preserving minors by taking $P$ to be the pair-set of all terminals. 

Furthermore, by exploiting a tight integration of our techniques with the compact distance oracles for planar graphs by Thorup~\cite{Thorup04}, we can show the following theorem regarding the size of reachability-preserving minors for planar digraphs. 

\begin{theorem} \label{thm: ubPlanar}
	Given a $k$-terminal planar digraph $G$, there exists a reachability-preserving minor $H$ of $G$ with size $O(k^{2} \log k)$.
\end{theorem}

We complement the above result by showing that there exist instances where the above upper-bound is tight up to a $O(\log k)$ factor. 


\begin{theorem} \label{thm: lbPlanar}
	For infinitely many $k \in \mathbb{N}$ there exists a $k$-terminal acyclic directed grid $G$ such that any reachability-preserving minor of $G$ must use $\Omega(k^{2})$ non-terminals.
\end{theorem}

\paragraph*{Cut, Flow and Distance Sparsifiers.} We provide new constructions for quality-$1$ (exact) cut, flow and distance sparsifiers for $k$-terminal planar graphs, where all the terminals are assumed to lie on the same face. We call such $k$-terminal planar graphs \emph{Okamura-Seymour} (OS) instances. They are of particular interest in the algorithm design and optimization community, due to the classical Okamura-Seymour theorem that characterizes the existence of feasible concurrent flows in such graphs~(see e.g., \cite{OkamuraS81,chekuri09,chekuri2010flow,LeeMM13}).

We show that the size of quality-$1$ sparsifiers can be as small as $O(k^2)$ for such instances, for which only exponential (in $k$) size of cut and flow sparisifiers were known before~\cite{KrauthgamerR13,andoni}. Formally, we have the following theorem.
\begin{theorem} \label{thm: mainThm}
	For any \emph{OS} instance $G$, i.e., a $k$-terminal planar graph in which all terminals lie on the same face, there exist quality-$1$ cut, flow and distance sparsifers of size $O(k^2)$. Furthermore, the resulting sparsifiers are also planar.
\end{theorem} 

We remark that all the above sparsifiers can be constructed in polynomial time (in $n$ and $k$), but we will not optimize the running time here. As we mentioned above, previously the only known upper bound on the size of quality-$1$ cut and flow sparsifiers for OS instance was $O(k^22^{2k})$, given by~\cite{KrauthgamerR13,andoni}. Our upper bound for cut sparsifier also matches the lower bound of $\Omega(k^2)$ for OS instance given by~\cite{KrauthgamerR13}. More specifically, in \cite{KrauthgamerR13}, an OS instance (that is a grid in which all terminals lie on the boundary) is constructed, and used to show that any mimicking network for this instance needs $\Omega(k^2)$ edges, which is thus a lower bound for planar graphs (see the table below for an overview). Note that that even though our distance sparsifier is not necessarily a minor of the original graph $G$, it still shares the nice property of being planar as $G$. Furthermore, Krauthgamer and Zondiner~\cite{KrauthgamerZ12} proved that there exists a $k$-terminal planar graph $G$ (not necessarily an OS instance), such that any quality-$1$ distance sparsifier of $G$ that is planar requires at least $\Omega(k^2)$ vertices. 

\vspace{-0.2cm}
\renewcommand{\arraystretch}{1.1}
\begin{center}
\begin{tabular}{l|l|l|l}
Type of sparsifier & Graph family &  Upper Bound &  Lower Bound  \\ \hline \hline
Cut & Planar & $O(k2^{2k})$~\cite{KrauthgamerR13} & $|E(G')| \geq \Omega(2^{k})$~\cite{karpov2017exponential} 
\\ \hline
Cut & Planar OS & $O(k^2)$~[\textbf{new}] & $|E(G')| \geq \Omega(k^{2})$~\cite{KrauthgamerR13} \\ \hline \hline
Flow & Planar OS & $O(k^{2}2^{2k})$~\cite{andoni} & follows from cut  \\ \hline
Flow & Planar OS & $O(k^2)$~[\textbf{new}] & follows from cut \\ \hline \hline
Distance (minor) & Planar OS & $O(k^{4})$~\cite{distancepreserving} & $\Omega(k^2)$~\cite{distancepreserving}   \\ \hline
Distance (planar) & Planar OS & $O(k^{2})$~[\textbf{new}] &  \\ \hline
\end{tabular}
\end{center}

We further provide a lower bound on the size of any \emph{data structure} (not necessarily a graph) that approximately preserves pairwise terminal distances of \emph{general} $k$-terminal graphs, which gives a trade-off between the distance stretch and the space complexity. 
 
\begin{theorem}~\label{thm:incompressibility}
For any $\varepsilon > 0$ and $t \geq 2$, there exists a (sparse) $k$-terminal $n$-vertex graph such that $k=o(n)$, and any data structure that approximates pairwise terminal distances within a multiplicative factor of $t - \varepsilon$ or an additive error $2t-3$ must use $\Omega(k^{1+1/(t-1)})$ bits. 
\end{theorem}

\textbf{Remark.}~ Recently and independently of our work, Krauthgamer and Rika~\cite{krauthgamer2017refined} constructed quality-$1$ cut sparsifiers of size $O(\gamma^{5}2^{2\gamma}k^4)$ for planar graphs whose terminals are incident to at most $\gamma=\gamma(G)$ faces. In comparison with our upper-bound which only considers the case $\gamma = 1$, the size of our sparsifiers from Theorem~\ref{thm: mainThm} is better by a $\Omega(k^{2})$ factor. Subsequent to our work, Karpov et al.~\cite{karpov2017exponential} proved that there exists edge-weighted $k$-terminal planar graphs that require $\Omega(2^k)$ edges in any exact cut sparsifier, which implies that it is necessary to have some additional assumption (e.g., $\gamma=O(1)$) to obtain a cut sparsifier of $k^{O(1)}$ size. 

\subsection{Our Techniques}
\label{subsec:techniques}
Our results for reachability-preserving minors are obtained by exploiting a technique of counting ``branching'' events between shortest paths in the directed setting (this technique was introduced by Coppersmith and Elkin~\cite{coppersmithE06}, and has also been recently leveraged by Bodwin~\cite{Bodwin17} and Abboud and Bodwin~\cite{AB18reachability}). Using this and a consistent tie-breaking scheme for shortest paths, we can efficiently construct a RPM for general digraphs of size $O(k^4)$ and by using a more refined analysis of branching events (see \cite{AB18reachability}), we can further reduce the size to be $O(k^3)$. We then combine our construction with a decomposition for planar digraphs (see~\cite{Thorup04}), to show that it suffices to maintain the reachability information among $O(k\log k)$ terminal pairs, instead of the naive $O(k^2)$ pairs, and then construct a RPM for planar digraphs with $O(k^2\log k)$ vertices. The lower-bound follows by constructing a special class of $k$-terminal directed grids and showing that any RPM for such grids must use $\Omega(k^2)$ vertices. Similar ideas for proving the lower bound on the size of distance-preserving minors for undirected graphs have been used by Krauthgamer et al.~\cite{distancepreserving}.

We construct our quality-$1$ cut and distance sparsifiers by repeatedly performing \emph{Wye-Delta transformations}, which are local operations that preserve cut values and distances and have proven very powerful in analyzing electrical networks and in the theory of circular planar graphs~(see e.g., \cite{curtis98,feo1993}). Khan and Raghavendra~\cite{KhanR14} used Wye-Delta transformations to construct quality-$1$ cut sparsifiers of size $O(k)$ for trees and outerplanar graphs, while our case (i.e., the planar OS instances) is more general and complicated and previously it was not clear at all how to apply such transformations to a more broad class of graphs. Our approach is as follows. Given a $k$-terminal planar graph with terminals lying on the same face, we first embed it into some large grid with terminals lying on the boundary of the grid. Next, we show how to embed this grid into a ``more suitable'' graph, which we will refer to as ``half-grid''. Finally, using the Wye-Delta operations, we reduce the ``half-grid'' into another graph whose number of vertices can be bounded by $O(k^{2})$. Since we  argue that the above graph reductions preserve exactly all terminal minimum cuts, our result follows. Gitler~\cite{Gitler91} proposed a similar approach for studying the reducibility of multi-terminal graphs with the goal to classify all Wye-Delta reducible graphs, which is very different from our motivation of constructing small vertex sparsifiers with good quality. 

The distance sparsifiers can be constructed similarly by slightly modifying the Wye-Delta operation. Our flow sparsifiers follow from the construction of cut sparsifiers and the flow/cut gaps for OS instances (which has been initially observed by Andoni et al.~\cite{andoni}). Our lower bound on the space complexity of any compression function approximately preserving terminal pairwise distance is derived by combining extremal combinatorics construction of Steiner Triple System that was used to prove lower bounds on the size of distance approximating minors (see~\cite{cheung2016}) and the incompressibility technique from~\cite{matousek96}.

\subsection{Related Work}
There has been a long line of work on investigating the tradeoff between the quality of the vertex sparsifier and its size~(see e.g., \cite{englert10,KrauthgamerR13,andoni} and Section~\ref{subsec:techniques}). (Throughout, cut, flow and distance sparsifiers will refer to their vertex versions.) Quality-$1$ \emph{cut sparsifiers} (or equivalently, mimicking networks) were first introduced by Hagerup et al.~\cite{HagerupKNR98}, who proved that for any graph $G$, there always exists a mimicking network of size $O(2^{2^k})$. 
Krauthgamer and Rika~\cite{KrauthgamerR13} showed how to build a mimicking network of size $O(k^22^{2k})$ for any planar graph $G$ that is minor of the input graph. They also proved a lower bound of $\Omega(k^2)$ on the number of edges of the mimicking network of planar graphs, and a lower bound of $2^{\Omega(k)}$ on the number of vertices of the mimicking network for general graphs. 

Quality-$1$ vertex flow sparsifiers have been studied in~\cite{andoni, GoranciR16}, albeit only for restricted families of graphs like quasi-bipartite, series-parallel, etc. It is not known if any general undirected graph $G$ admits a constant quality flow sparsifier with size independent of $|V(G)|$ and the edge capacities. For the quality-$1$ distance sparsifiers, Krauthgamer et al.~\cite{distancepreserving} introduced the notion of \emph{distance-preserving minors}, and showed an upper-bound of size $O(k^4)$ for general undirected graphs. They also gave a lower bound of $\Omega(k^2)$ on the size of such a minor for planar graphs. Recently, Abboud et al.~\cite{AbboudGMW17} show how to compress a planar graph metric using only $\tilde{O}(\min\{k^2,\sqrt{k \cdot n}\})$ bits. However, it is not clear whether their compression scheme can be represented by a graph on $k$ terminals which matches their bound. 

Over the last two decades, there has been a considerable amount of work on understanding the tradeoff between the sparsifier's quality $q$ and its size for $q>1$, i.e., when the sparsifiers only \emph{approximately} preserve the corresponding properties~\cite{juliasteiner,andoni,moitra09,leighton,charikar,englert10,mm10,gupta01,chekuri2006embedding,chan06,englert10,kamma2015,cheung2016, Cheung17,Filtser17,GajjarR17,DaubelDKMS17}. 

\section{Preliminaries} \label{sec: preli}
Let $G=(V,E)$ be a directed graph with terminal set $K \subset V$, $|K|=k$, which we will refer to as a \emph{$k$-terminal digraph}. We say $G$ is a \emph{$k$-terminal} DAG if $G$ has no directed cycles. The \emph{in-degree} of a vertex $v$, denoted by $\deg_{G}^{-}(v)$, is the number of edges directed towards $v$ in $G$. A digraph $H=(V_H,E_H)$, $K \subset V_H$ is a (\emph{vertex}) \emph{reachability sparsifier} of $G$ if for any $x ,x' \in K$, there is a directed path from $x$ to $x'$ in $H$ iff there is a directed path from $x$ to $x'$ in $G$. If $H$ is obtained by performing minor operations in $G$, then we say that $H$ is a \emph{reachability-preserving minor} of $G$. We define the \emph{size} of $H$ to be the number of non-terminals in $H$, i.e. $|V_H \setminus K|$.

Let $G=(V,E,c)$ be an undirected graph with terminal set $K \subset V$ of
cardinality $k$, where $c: E \rightarrow \mathbb{R}_{\geq 0}$ assigns a non-negative
capacity to each edge. We will refer to such a graph as a \emph{k-terminal graph}. 
Let $U \subset V$ and $S \subset K$. We say that a cut $(U, V \setminus U)$ is
$S$-separating if it separates the terminal subset $S$ from its complement $K
\setminus S$, i.e., $U \cap K$ is either $S$ or $K \setminus S$. We will refer to such cut as a \emph{terminal cut}. The cutset
$\delta(U)$ of a cut $(U, V \setminus U)$ represents the edges that have one
endpoint in $U$ and the other one in $V \setminus U$. The cost
$\capacity_G(\delta(U))$ of a cut $(U, V \setminus U)$ is the sum over all
capacities of the edges belonging to the cutset. We let $\text{mincut}_{G}(S, K
\setminus S)$ denote the minimum cost of any $S$-separating cut of
$G$. A graph $H = (V_H, E_H, c_H)$, $K \subset V_H$ is a \emph{quality-$q$} (\emph{vertex}) \emph{cut sparsifier} of $G$ with $q \geq 1$ if for any $S \subset K,
~ \mincut_G(S, K \setminus S) \leq \mincut_H(S, K \setminus S) \leq q \cdot
\mincut_G(S, K \setminus S).$

\section{Reachability-Preserving Minors for General Digraphs} \label{sec: minorsGeneralDigraphs}
In this section, we provide two constructions for reachability-preserving minors for general digraphs. The resulting minor from the first construction has size $O(k^4)$, which is larger than the size $O(k^3)$ of the minor from the second construction. However, our first construction can be implemented in polynomial time (in $n$), while the second one requires exponential running time.
\subsection{A Warm-up: An Upper Bound of $O(k^4)$}
In this section we show that any $k$-terminal digraph admits a reachability-preserving minor of size $O(k^{4})$. We accomplish this by first restricting our attention to DAGs, and then showing how to generalize the result to any digraph. 


We start by introducing the following definition. Given a digraph $G$ with a terminal set $K$ of size $k$ and a pair-set $P\subseteq K\times K$, we say that $H$ is a reachability-preserving minor with respect to $P$, if $H$ is a minor of $G$ that preserves the reachability information only among the pairs in $P$. Note that in the definition of vertex reachability sparsifiers, the \emph{trivial} pair-set $P$ contains $k(k-1)$ terminal-pairs, i.e., for any pair $x,x' \in K$, both $(x,x')$ and $(x',x)$ belong to $P$. Whenever we omit $P$, we mean to preserve the reachability information among all possible terminal pairs.

We next review a useful scheme for breaking ties between shortest paths connecting some vertex pair from $P$. This tie-breaking is usually achieved by slightly perturbing the edge lengths of the original graph such that no two paths have the same length (note that in our case, edge lengths are initially one). The perturbation gives a \emph{consistent} scheme in the sense that whenever $\pi$ is chosen as a shortest path, every sub-path of $\pi$ is also chosen as a shortest path. Below we formalize these ideas using two definitions and a lemma from~\cite{Bodwin17}.

\begin{definition}[Tie-breaking Scheme] Given a $k$-terminal $G$, a \emph{shortest path tie breaking scheme} is a function $\pi$ that maps every pair of vertices $(s,t)$ to some shortest path between $s$ and $t$ in $G$. For any pair-set $P$, we let $\pi(P)$ denote the union over all shortest paths between pairs in $P$ with respect to the scheme $\pi$.
\end{definition}

\begin{definition}[Consistency] A tie-breaking scheme is consistent if, for all vertices $y,x,x',y' \in V$, if $x,x' \in \pi(y,y')$ with $d(y,x) < d(y,x')$, then $\pi(x,x')$ is a sub-path of $\pi(y,y')$.
\end{definition}


\begin{lemma}[\cite{Bodwin17}] \label{lemm: consistency} For any $k$-terminal G, there is a consistent tie-breaking scheme in $G$.
\end{lemma}

We remark that for any $k$-terminal graph with $n$ vertices, the consistent tie-breaking scheme can be constructed in polynomial (in $n$) time~\cite{coppersmithE06}. 

Let $G$ be a $k$-terminal DAG. Given a tie-breaking scheme $\pi$, the first step to construct a reachability-preserving minor is to start with an empty graph $H$ and then for every pair $p \in P$, repeatedly add the shortest-path $\pi(p)$ to $H$. We can alternatively think of this as deleting vertices and edges that do not participate in any shortest path among terminal-pairs in $P$ with respect to the scheme $\pi$. Clearly, the DAG $H=(V_H, E_H)$, $E_H := \pi(P)$, is a minor of $G$ and preserves all reachability information among pairs in $P$. We next review the notion of a branching event, which will be useful to bound the size of $H$.  

\begin{definition}[Branching Event] A \emph{branching event} is a set of two distinct directed edges $\{e_1=(u_1,v), e_2=(u_2,v)\}$ that enter the same node $v$.
\end{definition}  

\begin{lemma} \label{lemm: branching} The \emph{DAG} $H$ has at most $|P|(|P|-1|)/2$ branching events.
\end{lemma}
\begin{proof}
	First, note that by construction of $H$, we can associate each edge $e \in E_H$ with some pair $p \in P$ such that $e \in \pi(p)$. To prove the lemma, it suffices to show that for any two terminal-pairs $p_1, p_2 \in P$, there is at most one branching event in the graph induced by $\pi(p1) \cup \pi(p_2)$. Suppose towards contradiction that there exist two terminal pairs $p_1, p_2$ that have two branching events in $\pi(p1) \cup \pi(p2)$. More specifically, we assume there exist two branching events
	\[
	b := \{e_1 = (u_1,v), e_2=(u_2,v)\} \text{ and } b' := \{e_1 = (u_1',v'), e_2=(u_2',v')\},
	\]
	where $e_i$ and $e_i'$ lie on the dipath $\pi(p_i)$, for $i = 1,2$.
	
	Assume without loss of generality that the vertex $v$ appears before $v'$ in the dipath $\pi(p_1)$. We then claim that $v$ must also appear before $v'$ in the dipath $\pi(p_2)$, since otherwise we would have a directed cycle between $v$ and $v'$, thus contradicting the fact that $H$ is acyclic. Since the tie-breaking scheme $\pi$ is consistent (Lemma \ref{lemm: consistency}), it follows that the dipaths $\pi(p_1)$ and $\pi(p_2)$ must share the subpath $\pi(v,v')$. Thus, $\pi(p_1)$ and $\pi(p_2)$ use the same edge that enters the node $v'$, i.e., $e_1' = e_2'$. However, by definition of a branching event, the edges that enter a node must be distinct, contradicting the fact that $b'$ is a branching event. This implies that there cannot be two branching events for the terminal pairs $p_1$ and $p_2$, thus proving the lemma.  
\end{proof}
We now have all the tools to present our algorithm for constructing reachability-preserving minors for DAGs. 
\vspace{-0.2cm}
\begin{algorithm}[H]
\caption{\textsc{MinorSparsifyDag} ($k$-terminal DAG $G$, pair-set $P$)}
\begin{algorithmic}[1]
\State Set $H = \emptyset$ and compute a consistent tie-breaking scheme $\pi$ for shortest paths in $G$.
\State\label{alg:minor_2} For each $p \in P$,  add the shortest path $\pi(p)$ to $H$.
\While{there is an edge $(u,v)$ directed towards a non-terminal $v$ with $\deg_{H}^{-}(v) = 1$}
\State\label{alg:minor_5} Contract the edge $(u,v)$.
\EndWhile
\State \Return $H$
\end{algorithmic}
\label{algo: minor}
\end{algorithm}
\vspace{-0.3cm}


\begin{lemma} \label{lemma: DAG}
Given a $k$-terminal \emph{DAG} $G$ with a pair-set $P$, \emph{Algorithm~\ref{algo: minor}} outputs a reachability-preserving minor $H$ of size $O(|P|^2)$ for $G$ with respect to $P$.
\end{lemma}
\begin{proof}
	We first argue that $H$ is a reachability-preserving minor with respect to the terminals. Indeed, after Line \ref{alg:minor_2} of the algorithm, graph $H$ can viewed as deleting vertices and edges from $G$ that do not lie on any of the shortest path among terminal pairs in $P$, chosen according to the scheme $\pi$. Thus, at this point $H$ is clearly a minor of $G$ that preserves the reachability information among the pairs in $P$. The edge contractions we perform in the remaining part of the algorithm guarantee that the resulting $H$ remains a reachability-preserving minor of $G$ with respect to $P$. 
	
	To bound the size of $H$, note that every non-terminal $v \in V_H \setminus K$ has in-degree at least $2$, and thus it corresponds to at least one branching event. Lemma \ref{lemm: branching} shows that the number of branching events is at most $O(|P|^2)$. Observing that edge contractions in Line \ref{alg:minor_5} do not affect this number, 
	we get that the size of $H$ is $O(|P|^2)$.
\end{proof}
%


We next show how the construction of reachability-preserving minors can be reduced from general digraphs to DAGs, and prove the following theorem.

\begin{theorem} \label{thmi: general1}
	Given a $k$-terminal digraph G with a pair-set $P$, there exists a polynomial-time algorithm that outputs a reachability-preserving minor $H$ of size $O(|P|^2)$ with respect to $P$.
\end{theorem}
Taking $P$ to be the trivial pair-set, we get a reachability-preserving minor of size $O(k^{4})$. 

\begin{proof}[Proof of Theorem~\ref{thmi: general1}]
Recall that a digraph is \emph{strongly connected} if there is a directed path between all pair of vertices. We proceed by first finding a decomposition of the graph into strongly connected components (SCCs)~\cite{Tarjan72}. We observe that each SCC that contains terminals can be contracted into a smaller component only on the terminals. Then contracting each SCC into a single vertex to obtain a DAG. Invoking Algorithm~\ref{algo: minor} on the resulting DAG gives some intermediate reachability-preserving minor. Finally, we show that this minor can be expanded back to produce a reachability-preserving minor for the original digraph. These steps are formally given in the procedure Algorithm~\ref{algo: minorSparsify}.

\begin{algorithm}[H]
	\caption{\textsc{MinorSparsify} ($k$-terminal digraph $G$, pair-set $P$)}
	\begin{algorithmic}[1]
		\State Compute a strongly connected component decomposition $\mathcal{D}$ of $G$.
		\State Let $f$ be some initially empty labelling that records the SCC of every vertex.
		\ForAll{SCC $C \in \mathcal{D}$} 
		\If{$C$ contains some terminal $x \in K$}
		\State For all $v \in C$, $f(v) = x$.
		\Else~{Choose some arbitrary $u \in C$, and set $f(v) = u$, for all $v \in C$.}
		\EndIf
		\EndFor
		\State \texttt{// Preprocessing Step}
		\State Let $\mathcal{D}_{K}$ denote the set of SCCs containing terminals in $G$.
		\ForAll{SSC $C \in \mathcal{D}_{K}$} 
		\While{$C$ contains some non-terminal $v$}
		\State\label{alg:ms_13} Choose some directed edge $(v,u)$ leaving $v$ inside $C$, and contract $v$ into $u$. 
		\EndWhile
		\EndFor
		\State Let $\hat{G}=(\hat{V}, \hat{E})$ and $\hat{\mathcal{D}}$ denote the resulting graph and the SCC decomposition.
		\State \texttt{// Main Procedure}
		\State Contract each SSC in $\hat{\mathcal{D}}$ into a single vertex, producing the DAG $G'=(V',E')$.
		\State Let $K'=\emptyset$ and $P'=\emptyset$ be the terminal set and pair-set of $G'$, respectively. 
		\State For all $k \in K$, add $f(k)$ to $K'$ and remove duplicates, if any.
		\State For all $(s,t) \in P$, add $(f(s),f(t))$ to $P'$ if $f(s) \neq f(t)$.
		\State Set $H' =$\textsc{MinorSparsifyDag}($G',P'$).
		\State\label{alg:ms_23} Let $H$ be the graph obtained by expanding back all contracted SCCs in $\hat{\mathcal{D}}_{K}$ in $H'$.
		
		\State \Return $H$
	\end{algorithmic}
	\label{algo: minorSparsify}
\end{algorithm}

The main intuition behind the correctness of the above reduction lies on two important observations. First, vertices belonging to the same strongly connected components can always reach each other. Second, vertices belonging to different strongly connected components can reach each other if the corresponding vertices in the contracted graph can do so. 
We have the following useful observation.

\begin{fact} \label{lemma: edgecontraction}
	For any strongly connected digraph $G=(V,E)$, contracting any edge $e \in E$ results in another strongly connected digraph $G'=(V',E')$.
\end{fact}

Now we show that the graph $H$ output by \textsc{MinorSparsify} is a reachability-preserving minor of $G$. It is easy to verify that the produced graph $H$ is indeed a minor of $G$. To show the correctness, we will prove that $H$ preserves the reachability information among all pairs from $P$ in $G$. Before doing that, observe that the graph $\hat{G}$ obtained after the preprocessing step is a reachability preserving minor of $G$ with respect to $P$. Indeed, this can be inferred by a repeated application of Fact~\ref{lemma: edgecontraction} to each SSC containing terminal vertices.
	
	Now, let $(s,t) \in P$ be any terminal-pair in $G$. Assume that $t$ is reachable from $s$ in $G$. We distinguish two cases:
	\begin{enumerate}
		\item If $s$ and $t$ belong to the same SCC in $\mathcal{D}$, they do also belong to the corresponding SCC in $\hat{\mathcal{D}}$. In Line \ref{alg:ms_13}, $s$ and $t$ are contracted into a single terminal. However, since the contracted SSC contains terminals, it is expanded back to its original form in $\hat{\mathcal{D}}$ in Line \ref{alg:ms_23}. Thus, it follows that $t$ is reachable from $s$ in the output graph $H$. 
		\item If $s$ and $t$ do not belong to the same SCC in $\mathcal{D}$, they must also not belong to the same SCC in $\hat{\mathcal{D}}$. Let $f(s)$ and $f(t)$ denote the terminals in the DAG $G'$ obtained by contracting their corresponding components in $\hat{\mathcal{D}}$ (Line \ref{alg:ms_13}). Since $t$ is reachable from $s$ in $\hat{G}$, note that $f(t)$ must also be reachable from $f(s)$ in $G'$. By Lemma \ref{lemma: DAG}, it follows that $f(t)$ is reachable from $f(s)$ in the reachability-preserving minor $H'$ of $G'$. Expanding back the SCCs that contain terminals in $H'$ (Line \ref{alg:ms_23}), we can construct the directed path $s \rightsquigarrow f(s) \rightsquigarrow f(t) \rightsquigarrow t$ in $H$, which shows that $t$ is also reachable from $s$ in the output graph $H$.
	\end{enumerate}
	When $t$ is not reachable from $s$ in $G$, we can similarly show that $t$ is also not reachable from $s$ in $H$, thus concluding the correctness proof.
	
	We now bound the size of $H$. Since the DAG $G'$ has $|P'| \leq |P|$ pairs, it follows by Lemma \ref{lemma: DAG} that $H'$ has size at most $O(|P|^2)$. After expanding back the SCCs in Line 19, we get that each SSC in $H$ contains at most $k_i$ terminals, where $k = \sum_{i} k_i$. Note that this does not contribute to the size of $H$. Therefore, we get that the size of the output graph $H$ is at most $O(|P|^2)$.
\end{proof}

\subsection{An Improved Bound of $O(k^3)$}
Using the recent work due to Abboud and Bodwin~\cite{AB18reachability}, we next show how to get a polynomial improvement on the number of branching events from Lemma~\ref{lemm: branching}. This in turn gives a polynomial improvement on the size of reachability-preserving minor from Theorem~\ref{thmi: general1}.

Specifically, given a $k$-terminal DAG $G$ with a pair-set $P$, let $H=(V,E_H)$ be the subgraph of $G$ with minimum number of edges that preserves all reachability information among the pairs in $P$. We call such an $H$ the \emph{sparsest reachability preserver} of $G$.  The following lemma is implicit in~\cite{AB18reachability}, and we include it here for the sake of completeness.

\begin{lemma} The \emph{DAG} $H=(V,E_H)$ has at most $k \cdot |P|$ branching events.
\end{lemma}
\begin{proof}
For each pair $(s,t) \in P$, we associate a directed path $s \rightsquigarrow t$ in $H$, and let $\tilde{\pi}(s,t)$ denote such a path. Note that since $H$ is acyclic, every $\tilde{\pi}(s,t)$ is acyclic as well. Moreover, using the fact that $H$ is the sparsest reachability preserver, it follows that for every edge $e \in E_H$, there must be some pair $(s,t) \in P$ such that deleting $e$ from $H$ implies that $s$ cannot reach $t$, i.e., $s \not \rightsquigarrow t$ in $H \setminus \{e\}$. This naturally leads to a relationship between edges and pairs. Specifically, we say that every edge $e \in E_H$ is \emph{owned} by one such pair $(s,t) \in P$. 

Next, for each $(s,t) \in P$, we let $B^{H}_{(s,t)}$ denote the set of all branching events $\{e_1,e_2\}$ in $H$ such that either $e_1$ or $e_2$ (but not both) is owned by $(s,t)$. We claim that $\bigcup \set{B^{H}_{(s,t)}}{(s,t) \in P}$ contains all branching events in $H$. Indeed, suppose towards contradiction that $\{e_1,e_2\}$ is a branching event in $H$ but not in $\bigcup \set{B^{H}_{(s,t)}}{(s,t) \in P}$. Then by definition of $B^{H}_{(s,t)}$ there must be some pair $(s,t) \in P$ such that $e_1$ and $e_2$ are both owned by $(s,t)$. The latter implies that we can construct two directed paths from $s$ to $t$, where one path uses $e_1$ and the other uses $e_2$. Delete edge $e_1$ w.l.o.g. Then we still have another directed path from $s$ to $t$, thus contradicting the assumption that $e_1$ is owned by $(s,t)$. 

Now, to prove the lemma it suffices to show that $\abs{B^{H}_{(s,t)}} \leq k$, for every $(s,t) \in P$. Suppose towards contradiction that there exists a pair $(s,t) \in P$ such that $\abs{B^{H}_{(s,t)}} \geq k+1$.  Then by the pigeonhole principle, there exist two branching events 
\[
	\{(x_1,b_1),(x_2,b_1)\},~\{(y_1,b_2),(y_2,b_2)\} \in B^{H}_{(s,t)}
\]
entering the nodes $b_1$ and $b_2$, such that $(s,t)$ owns $(x_1,b_1)$ and $(y_1,b_2)$, and the other edges are owned by pairs that share a common left terminal, i.e., 
\[
	(x_2,b_1) \text{ is owned by } (u,v_1) \text{ and } (y_2,b_2) \text{ is owned by } (u,v_2)
\]
for some $u \in K$ and $(u,v_1), (u,v_2) \in P$. Note that by the definition of $B^H_{(s,t)}$, $y_1$ and $y_2$ are different vertices. We further assume w.l.o.g. that node $b_1$ appears before $b_2$ in $\tilde{\pi}(s,t)$.
Now, since the pair $(u,v_2)$ owns the edge $(y_2,b_2)$, every path $u \rightsquigarrow v_2$ must use the edge $(y_2,b_2)$, which further implies that every path $u \rightsquigarrow b_2$ must use the edge $(y_2,b_2)$. We can form a path $u \rightsquigarrow b_2$ by first taking the path $\tilde{\pi}(u,v_1)[u \rightsquigarrow b_1]$\footnote{Let $x,y,x',y' \in V$, $\tilde{\pi}(x,y)$ be a directed path from $x$ to $y$, and suppose $x',y' \in \tilde{\pi}(x,y)$ with $x'$ appearing before $y'$. Then $\tilde{\pi}(x,y)[x' \rightsquigarrow y']$ denotes the directed subpath from $x'$ to $y'$ in $\tilde{\pi}(x,y)$.} and then extend it by concatenating it with the path $\tilde{\pi}(s,t)[b_1 \rightsquigarrow b_2]$. This implies one of the following cases: (1) $(y_2,b_2) \in \tilde{\pi}(s,t)[b_1 \rightsquigarrow b_2]$ or (2) $(y_2,b_2) \in \tilde{\pi}(u,v_1)[u,b_1]$. We show that (2) cannot happen, thus only (1) holds. To this end, suppose towards contradiction that $(y_2,b_2) \in \tilde{\pi}(u,v_1)[u,b_1]$. Then we can find a directed path $b_2 \rightsquigarrow b_1$. But since $b_1$ appears before $b_2$, we get the cycle $b_2 \rightsquigarrow b_1 \rightsquigarrow b_2$, which contradicts the fact that $H$ is acyclic. 

Finally, case (1) implies that $(y_2,b_2) \in \tilde{\pi}(s,t)$. Therefore, the path $\tilde{\pi}(s,t)$ contains both $(y_1,b_2)$ and $(y_2,b_2)$. On the other hand, since $\tilde{\pi}(s,t)$ is acyclic, there cannot be two vertices entering $b_2$, which is a contradiction.  
\end{proof}

The above lemma leads to the following algorithm.

\begin{algorithm}[H]
\caption{\textsc{MinorSparsifyDag2} ($k$-terminal DAG $G$, pair-set $P$)}
\begin{algorithmic}[1]
\State Set $H=(V,E_H)$ be the sparsest reachability preserver with respect to $P$.
\State Remove isolated non-terminal vertices from $H$, if any.
\While{there is an edge $(u,v)$ directed towards a non-terminal $v$ with $\deg_{H}^{-}(v) = 1$}
\State Contract the edge $(u,v)$.
\EndWhile
\State \Return $H$
\end{algorithmic}
\label{algo: minor2}
\end{algorithm}

We remark that the above construction is built upon the sparest reachability preserver $H$, which we can find in exponential time (say, by a brute-force approach). By using similar arguments as in the proof of Lemma~\ref{lemma: DAG} and Theorem~\ref{thmi: general1}, we have the following guarantees.
\begin{lemma} Given a $k$-terminal \emph{DAG} $G$ with a pair-set $P$, \emph{Algorithm \ref{algo: minor2}} outputs a reachability-preserving minor $H$ of size $O(k \cdot \abs{P})$ for $G$ with respect to $P$.
\end{lemma}

\begin{theorem} \label{thmi: general}
Given a $k$-terminal digraph G with a pair-set $P$, there exists an algorithm that outputs a reachability-preserving minor $H$ of size $O(k \cdot \abs{P})$ with respect to $P$.
\end{theorem}
Taking $P$ to be the trivial pair-set we get a reachability-preserving minor of size $O(k^{3})$, which proves Theorem~\ref{thmi: general_reachability}. We note that in contrast to Theorem~\ref{thmi: general1}, the above theorem guarantees only an exponential-time algorithm in the worst-case. As discussed above, this comes from the assumption that we have access to the sparsest reachability preserver. It is conceivable that a similar approach that appears in~\cite{AB18reachability} could be employed to achieve a better running-time. However, the focus of our paper is on optimizing the size of reachability-preserving minors.

\section{Reachability-Preserving Minors for Planar Digraphs}
In this section we show that any $k$-terminal planar digraph $G$ admits a reachability-preserving minor of size $O(k^{2} \log{k})$ and thus prove Theorem~\ref{thm: ubPlanar}. This matches the lower-bound of Theorem~\ref{thm: lbPlanar} up to an $O(\log {k})$ factor. The main idea is as follows. Given a $k$-terminal planar digraph $G$ with the trivial pair-set $P$, $|P|=k(k-1)$, our goal will be to slightly increase the number of terminals while considerably reducing the size of the pair-set $P$, under the condition that no reachability information is lost among the terminal-pairs in $P$.

\vspace{-0.3cm}
\paragraph*{Preprocessing Step.} Given a $k$-terminal digraph $G$, we apply Theorem~\ref{thmi: general_reachability} to get a reachability-preserving minor $G'$. To simplify the notation, we will use $G$ instead of $G'$, i.e., throughout we assume that $G$ has at most $O(k^{3})$ vertices. 
\vspace{-0.3cm}
\paragraph{Decomposition into Path-Separable Digraphs and the Algorithm.} We say that a graph $G=(V,E)$ admits an \emph{$\alpha$-separator} if there exists a set $S \subset V$ whose removal partitions $G$ into connected components, each of size at most $\alpha \cdot |V|$, where $1/2 \leq \alpha < 1$. If the vertices of $S$ consist of the union over $r$ paths of $G$, for some $r \geq 1$, we say that $G$ is $(\alpha, r)$-\emph{path separable}. We now review the following reduction due to Thorup~\cite{Thorup04}.
\begin{theorem}[\cite{Thorup04}] \label{thm: thorup}
Given a digraph $G$, we can construct a series of digraphs $G_0,\ldots,G_{b}$ for some $b\leq n$ such that the number of vertices and edges over all $G_i$'s is linear in the number of vertices and edges in $G$, and
\begin{enumerate}
	\setlength{\itemsep}{2pt}
	\setlength{\parskip}{2pt}
	\item Each vertex and edge of $G$ appears in at most two $G_i$'s.
	\item\label{item:thorup_2} For all $u,v \in V$, if there is a dipath $R$ from $u$ to $v$ in $G$, there is a $G_i$ that contains $R$.  
	\item\label{item:thorup_3} Each $G_i=(V_i,E_i)$ is $(1/2,6)$-path separable.
	\item\label{item:thorup_4} Each $G_i$ is a minor of $G$. In particular, if $G$ is planar, so is $G_i$.
\end{enumerate}
\end{theorem}

Now we review how directed reachability can be efficiently represented by separator dipaths. Let $G$ be a $k$-terminal directed graph $G$ that contains some directed path $Q$. Assume that the vertices of $Q$ are ordered in increasing order in the direction of $Q$. For each terminal $x \in K$, let $\texttt{to}_x[Q]$ be the first vertex in $Q$ that can be reached by $x$, and let $\texttt{from}_x[Q]$ be the last vertex in $Q$ that reaches $x$. Let $(s,t)$ be a terminal pair and let $R$ be the directed path from $s$ to $t$ in $G$. We say that $R$ intersects $Q$ iff $s$ can reach $\texttt{to}_s[Q]$  and $t$ can be reached from $\texttt{from}_t[Q]$ in $Q$, and $\texttt{to}_s[Q]$ precedes $\texttt{from}_t[Q]$ in $Q$. 

We now are going to combine the above tools to give our labelling algorithm aimed at reducing the size of the trivial pair-set $P$. By Theorem~\ref{thm: thorup}, we restrict our attention only to the digraphs $G_i$. Let $K_i := V(G_i) \cap K$ be the set of terminals restricted to the graph $G_i$.
\begin{algorithm}[t]
	\caption{\textsc{ReducePairSet} (planar digraph $G_i$, terminals $K_i$)}
	\label{alg:reducepairset}
	\begin{algorithmic}[1]
		\State \textbf{if }{$|V(G_i)| \leq 1$ or $K_i = \emptyset$ \textbf{then return} $\emptyset$.}
		\State Let $P_i' = \emptyset$ be the new pair-set. 
		\State Compute a $1/2$-separator $S$ of $G_i$ consisting of $6$ dipaths by Item \ref{item:thorup_4} of Theorem~\ref{thm: thorup}.
		\For{each dipath $Q \in S$}
		\State \texttt{// Addition of terminal connections with $Q$}
		\State Let $Q'$ be the set of existing terminals of $Q$.
		\For{each terminal $x \in K_i$}
		\State\label{alg:rps_8} Compute $\texttt{to}_x[Q]$ and $\texttt{from}_x[Q]$, declare them terminals and add them to $Q'$.
		\State\label{alg:rps_9} Add $(x, \texttt{to}_x[Q])$ and $(\texttt{from}_x[Q],x)$ to $P_i'$. 
		\EndFor
		\State \texttt{// Sparsification of $Q$ using $Q'$}
		\State\label{alg:rps_12} Define directed pairs $(s,t)$, where $s$ and $t$ are consecutive terminals of $Q'$, 
		\Statex according to the ordering of $Q$ and add all these pairs to $P_i'$.
		\EndFor
		\State Let $(G_i^{(1)},K_i^{(1)})$ and $(G_i^{(2)},K_i^{(2)})$ be the resulting graphs from $G \setminus S$,
		\Statex where $K_i^{(1)}$ and $K_i^{(2)}$ are disjoint subsets of the terminals $K$ separated by $S$.
		\State \texttt{// Note that reachability info. about terminals in $S$ are taken care of.}
		\State \Return $P_i' \cup \bigcup_{j=1}^{2}\textsc{ReducePairSet}(G_i^{(j)},K_i^{(j)})$.
	\end{algorithmic}
	\label{algo: reduceSize}
\end{algorithm}
\begin{lemma} Let $G$ be a $k$-terminal planar digraph. Let $P' := \cup_{i=0}^{b} P_i'$ be the union over all pair-sets output by running \emph{Algorithm \ref{alg:reducepairset}} on each digraph $G_i$. Then the size of $|P'|$ is at most $O(k \log k)$. Moreover, if $H$ is a reachability-preserving minor of $G$ with respect to $P'$, then $H$ is a reachability-preserving minor of $G$ with respect to all terminal pairs. 
\end{lemma} 
\begin{proof}
	By preprocessing, $G$ has at most $O(k^{3})$ vertices. Throughout, it will be useful to think of the above algorithm as simultaneously running it on each digraph $G_i$. By Item \ref{item:thorup_2} of Theorem~\ref{thm: thorup}, each terminal appears in at most two $G_i$'s. Thus at each recursive level, there will be at most $O(k)$ active $G_i$'s. Also, note that the separator properties imply that there are $O(\log k)$ recursive calls overall.
	
	We next bound the size of the pair-set $P'$. Let $q$ denote the total number of newly added terminals in Line \ref{alg:rps_8} per recursive level. Since there are $O(k)$ terminals, each adding at most $O(1)$ new terminals, it follows that $q=O(k)$. First, we argue about the number of pairs added in Line \ref{alg:rps_9}. Since this is bounded by $O(q)$, it follows that there are $O(k \log k)$ pairs overall. Second, we bound the number of pairs added when sparsifying the separator paths, i.e., pair additions in Line \ref{alg:rps_12}. For all the separators in the same recursive level, we can write $q := \sum_{i}{|Q'_j|}$, where $Q_j'$ denotes the set newly added terminals for some separator dipath. By Line \ref{alg:rps_12}, it follows that we need only $(|Q'_j| - 1)$ pairs to represent each such dipath. Thus, per recursive call, the total number of newly added pairs is $O(q) = O(k)$. Summing these overall $O(\log k)$ levels, and combining this with the previous bound, gives the claimed bound on $|P'|$.
	
	Finally, we argue that $P'$ is a pair-set that can recover reachability information among terminals. Fix any terminal pair $(s,t)$ and let $R$ be a directed path from $s$ to $t$ in $G$. By Item \ref{item:thorup_2} of Theorem~\ref{thm: thorup}, there is some digraph $G_i$ that contains $R$. Then, $R$ must intersect with some separator dipath $Q$, at some level of the recursion of the above algorithm on $G_i$. The above argument gives that $P'$ contains all the necessary information to give a (possibly) another directed path from $s$ to $t$ in $G$.
\end{proof}
Applying Theorem~\ref{thmi: general} on the digraph $G$ with pair-set $P'$, as defined by the above lemma, we get Theorem~\ref{thm: ubPlanar}.

\subsection{Reachability-Preserving Minors: Lower-bound for Planar DAGs} \label{app: minorLB}
In this section we prove that there exists an infinite family of $k$-terminal acyclic directed grids such that any reachability-preserving minor for such graphs needs $\Omega(k^2)$ non-terminals (i.e., prove Theorem~\ref{thm: lbPlanar}). We achieve this by adapting the ideas of Krauthgamer et al.~\cite{distancepreserving}, from their lower-bound proof on distance-preserving minors for undirected graphs.

We start by defining of our lower-bound instance. Fix $k$ such that $r=k/4$ is an integer. Construct an initially undirected $(r+1) \times (r+1)$ grid, where all the $k$ terminals lie on the boundary, except at the corners, and declare all non-boundary vertices non-terminals. Remove the four corner vertices, and then all boundary edges connecting the terminals. Now, make the graph directed by first directing each horizontal edge from left to right, and then directing each vertical edge from top to bottom. Let $G$ denote the resulting $k$-terminal directed grid. It is easy to verify that $G$ is acyclic. 

\begin{theorem}
	For infinitely many $k \in \mathbb{N}$ there exists a $k$-terminal acyclic directed grid $G$ such that any reachability-preserving minor of $G$ must use $\Omega(k^{2})$ non-terminals.
\end{theorem}
\begin{proof}
	Let $G$ be the $k$-terminal grid defined as above. Note that there are $r$ terminals on each side of the grid. Let $H$ be any reachability-preserving minor of $G$. Recall that $H$ contains all terminal vertices from $G$. Furthermore, let $x_1,x_2,\ldots,x_{r}$ be the terminals on the left-side of the grid, ordered from top to bottom. Similarly, define $y_1,y_2,\ldots,y_{r}$ to be the terminals on the right-side. Note that by construction of $G$, for an index pair $(i,j)$ with $i < j$, there is no directed path from $x_j$ to $y_i$. Finally, define $P^{i}_H$ to be the directed path from $x_i$ to $y_i$ in $H$, for $i=1,\ldots,r$. 
	Throughout we will refer to such paths as \emph{horizontal}.
	
	\begin{claim} \label{lemm: horizontal} The horizontal directed paths $P^{1}_H,P^{2}_H,\ldots,P^{r}_H$ are vertex disjoint in $H$.
	\end{claim}
	\begin{proof}
		Suppose towards contradiction that there exist some $i$ and $j$ with $i < j$ such that $P^{i}_H$ and $P^{j}_H$ intersect at some vertex $z$ in $H$. This implies that there are directed paths from $x_i$ and $x_j$ to $z$, and from $z$ to $y_i$ and $y_j$. The latter implies that there is a directed path from $x_j$ to $y_i$ in $H$. However, by construction of $G$, we know that $x_j$ cannot reach $y_i$ for $i < j$, contradicting the fact that $H$ is a reachability-preserving minor of $G$.
	\end{proof}
	We can apply symmetric argument to the \emph{vertical} paths in $H$. More specifically, define $u_1,u_2,\ldots,u_{r}$ to be the terminal on the top-side of the grid, order from left to right. Similarly, define $v_1,v_2,\ldots,u_{r}$ to be the terminals on the bottom-side. Note that by construction of $G$, for an index pair $(i,j)$ with $i<j$, there is no directed path from $u_j$ to $v_i$. Finally, define $Q^{i}_H$ to be the directed path from $u_i$ to $v_i$ in $H$, for $i=1,\ldots,r$. Then we get the following symmetric claim.
	\begin{claim} \label{lemm: vertical} The vertical directed paths $Q^{1}_H,Q^{2}_H,\ldots,Q^{r}_H$ are vertex disjoint in $H$.
	\end{claim} 
	We next argue that all the horizontal and the vertical paths must intersect with each other. 
	\begin{claim} \label{lemm: intersect} Any pair of horizontal and vertical paths $P^{i}_H$ and $Q^{j}_H$ intersect in $H$.
	\end{claim}
	\begin{proof}
		Since $H$ is a minor of $G$, any dipath that connects two terminals in $H$ can be mapped back to a dipath connecting two terminals in $G$. Let $P_i$ and $Q_j$ be the corresponding dipaths in $G$ that are obtained by expanding back the dipaths $P^{i}_H$ and $Q^{j}_H$ in $H$. By construction of $G$, the horizontal and vertical dipaths between terminals are unique, implying that $P_i$ and $Q_j$ must intersect at some vertex of $G$. By performing the backtracked minor-operations on this vertex yields an intersection vertex between $P^{i}_H$ and $Q^{j}_H$ in $H$.
	\end{proof}
	The last claim we need shows that no pair of horizontal and the vertical paths intersects intersect at a terminal vertex.
	\begin{claim} \label{lemm: noterminal}
		No pair of horizontal and vertical paths $P^{i}_H$ and $Q^{j}_H$ intersects at a terminal vertex in $G$.
	\end{claim}
	\begin{proof}
		Consider the terminal pairs $(x_i,y_i)$ and $(u_j,v_j)$ corresponding to the paths $P^{i}_H$ and $Q^{j}_H$. Note that by construction of $G$, the set of terminals reachable from both $x_i$ and $u_j$ in $G$ is $\{y_i, y_{i+1}, \ldots, y_{r}\} \cup \{v_j, v_{j+1}, \ldots, v_{r}\}$. Since $H$ is a reachability-preserving minor of $G$, $x_i$ and $u_j$ must also be able to reach this terminal-set in $H$ and also $P^{i}_H$ and $Q^{j}_H$ cannot intersect at any terminal in $\{y_1,\ldots,y_{i-1}\} \cup \{v_1,\ldots,v_{j-1}\}$. Now, suppose towards contradiction that $P^{i}_H$ and $Q^{j}_H$ intersect at some terminal $y_k$, for $k \in \{i+1,\ldots,r\}$. This implies that in the path $P^{i}_H$, there is a directed path from $y_k$ to $y_i$, for $k > i$, giving a contradiction by construction of $G$. Furthermore, observe that $P^{i}_H$ and $Q^{j}_H$ cannot intersect at $y_i$, as otherwise we would have a directed path from $y_i$ to $v_j$, which is a contradiction by construction of $G$. Applying a similar argument to the case when paths intersect at some terminal $v_\ell$, for $k \in \{j+1,\ldots,r\}$, gives the claim.
	\end{proof}
	We know have all the necessary tools to prove the theorem. Claim~\ref{lemm: intersect} shows that the paths $P^{i}_H$ and $Q^{j}_H$ intersect in $H$ and let $z_H^{i,j}$ denote one of the intersection vertices. Now, we must show that all these vertices are distinct. To this end, assume that $z_{H}^{i_1,j_1} = z_{H}^{i_2,j_2}$. Since these vertices belong to both $P^{i_1}_{H}$ and $P^{i_2}_{H}$, by Claim~\ref{lemm: horizontal} we get that $i_1 = i_2$. Similarly, by Claim~\ref{lemm: vertical} we get that $j_1 = j_2$. Thus, we have that all vertices $z_{H}^{i,j}$, for $i,j=1,2,\ldots, r$ are distinct. Since Claim~\ref{lemm: noterminal} implies that none of this intersection vertices is a terminal, we conclude that $H$ must contain at least $r^{2} = (k/4)^{2}$ non-terminals.
\end{proof}
\section{An Exact Cut Sparsifier of Size $O(k^2)$} \label{sec: upperCut}
In this section we show that given a $k$-terminal planar graph, where all terminals lie on the same face, one can construct a quality-$1$ cut sparsifier of size $O(k^{2})$. Note that it suffices to consider the case when all terminals lie on the \emph{outer} face. 
We first present some basic tools. 
\subsection{Basic Tools}
\paragraph*{Wye-Delta Transformations. }
In this section we investigate the applicability of some graph reduction techniques that aim at reducing the number of non-terminals in a $k$-terminal graph. 
We start by reviewing the so-called \emph{Wye-Delta} operations in graph reductions. These operations consist of five basic rules, which we describe below. (See Fig.~\ref{fig:wyedelta} for illustrations.)
\begin{enumerate}
	\setlength\itemsep{0.1em}
	\item \emph{Degree-one reduction:} Delete a degree-one non-terminal and its incident edge.
	\item \emph{Series reduction:} Delete a degree-two non-terminal $y$ and its incident edges $(x,y)$ and $(y,z)$, and add a new edge $(x,z)$ of capacity $\min\{c(x,y), c(y,z)\}$.
	\item \emph{Parallel reduction:} Replace all parallel edges by a single edge whose capacity is the sum over all capacities of parallel edges.
	\item \emph{Wye-Delta transformation:} Let $x$ be a degree-three non-terminal with neighbours $\delta(x) = \{u,v,w\}$. Assume w.l.o.g.\footnote{Suppose there exist a pair $(u,v) \in \delta(x)$ with $c(u,x) + c(v,x) < c(w,x)$, where $w \in \delta(v)\setminus \{u,v\}$. Then we can simply set $c(w,x) = c(u,x) + c(v,x)$, since any terminal minimum cut would cut the edges $(u,x)$ and $(v,x)$ instead of the edge $(w,x)$.} that for any pair $(u,v) \in \delta(x)$, $c(u,x) + c(v,x) \geq c(w,x)$, where $w \in \delta(v)\setminus \{u,v\}$. Then we can delete $x$ (along with all its incident edges) and add edges $(u,v),(v,w)$ and $(w,u)$ with capacities $(c(u,x)+c(v,x)-c(w,x))/2$, $(c(v,x)+c(w,x)-c(u,x))/2$ and $(c(u,x)+c(w,x)-c(v,x))/2$, respectively. 
	\item \emph{Delta-Wye transformation:} Delete the edges of a triangle connecting $x$, $y$ and $z$, introduce a new non-terminal vertex $w$ and add new edges $(w,x)$, $(w,y)$ and $(w,z)$ with edge capacities $c(x,y) + c(x,z),$ $c(x,y) + c(y,z)$ and $c(x,z) + c(y,z)$ respectively.
\end{enumerate}

\begin{figure}[H]
\centering
\scalebox{.8}{
\begin{tikzpicture}
\tikzstyle{vertex}=[circle,draw = white, fill=black, minimum size = 7pt, inner sep=2pt]
\tikzstyle{vertex1}=[fill = white, draw = white]
\tikzstyle{edge}=[-,thick ]
\tikzstyle{elipse}=[-,thick ]
\tikzstyle{vertex2}=[circle,draw = black, fill=white, minimum size = 7pt, inner sep=2pt]

  \node[vertex2] (x) at (-0.75,0.75) {$1$};

  \node[vertex1] (n1) at (0,0) {} ;
  \node[vertex1] (n2) at (0.5,0.5) {} ;
  \node[vertex1] (n3) at (0,1) {} ;
  
  \node[vertex] (n4) at (2.5,1.5) {};
    
  \node[vertex] (n5) at (2.5,0.5) {} ;
  \node[vertex1] (n6) at (2,0) {} ;
  \node[vertex1] (n7) at (3,0) {} ;

  \draw[edge] (n4) node[above = 2.5pt] {$y$} -- (n5)  ;

  \draw[edge] (n6) -- (n5) ;
  \draw[edge] (n5) -- (n7) ;

  \node[vertex1] (n11) at (4, 0.5) {};
  \node[vertex1] (n12) at (5, 0.5) {};
  \draw[edge,->,very thick,>=stealth] (n11) -- (n12) ;
  
  \node[vertex1] (n1a) at (6,0) {} ;
  \node[vertex] (n2a) at (6.5,0.5) {} ;
  \node[vertex1] (n3a) at (7,0) {} ;

  \node[vertex1] (n5a) at (8.5,0.5) {} ;
  \node[vertex1] (n6a) at (9,1) {} ;
  \node[vertex1] (n7a) at (9,0) {} ;

  \draw[edge] (n1a) -- (n2a) ;
  \draw[edge] (n2a) -- (n3a) ;
  


\end{tikzpicture}
}

\scalebox{.8}{

\begin{tikzpicture}
\tikzstyle{vertex}=[circle,draw = white, fill=black, minimum size = 7pt, inner sep=2pt]
\tikzstyle{vertex1}=[fill = white, draw = white]

\tikzstyle{edge}=[-,thick ]
\tikzstyle{elipse}=[-,thick ]
\tikzstyle{vertex2}=[circle,draw = black, fill=white, minimum size = 7pt, inner sep=2pt]

  \node[vertex2] (x) at (-0.75,0.5) {$2$};
  \node[vertex1] (n1) at (0,0) {} ;
  \node[vertex] (n2) at (0.5,0.5) {} ;
  \node[vertex1] (n3) at (0,1) {} ;
  
  \node[vertex] (n4) at (1.5,0.5) {};
    
  \node[vertex] (n5) at (2.5,0.5) {} ;
  \node[vertex1] (n6) at (3,1) {} ;
  \node[vertex1] (n7) at (3,0) {} ;
  
  \draw[edge] (n1) -- (n2) node[above = 2.5pt] {$x$} ;
  \draw[edge] (n2) -- (n3) ;
  
  \draw[edge] (n2) -- (n4) node[above = 2.5pt] {$y$} ;
  \draw[edge] (n4) -- (n5) node[above = 2.5pt] {$z$};

  \draw[edge] (n6) -- (n5) ;
  \draw[edge] (n5) -- (n7) ;

  \node[vertex1] (n11) at (4, 0.5) {};
  \node[vertex1] (n12) at (5, 0.5) {};
  \draw[edge,->,very thick,>=stealth] (n11) -- (n12) ;
  
  \node[vertex1] (n1a) at (6,0) {} ;
  \node[vertex] (n2a) at (6.5,0.5) {} ;
  \node[vertex1] (n3a) at (6,1) {} ;

  \node[vertex] (n5a) at (8.5,0.5) {} ;
  \node[vertex1] (n6a) at (9,1) {} ;
  \node[vertex1] (n7a) at (9,0) {} ;

  \draw[edge] (n1a) -- (n2a) ;
  \draw[edge] (n2a) -- (n3a) ;
  
  \draw[edge] (n2a) node[above = 2.5pt] {$x$} -- (n5a) node[above = 2.5pt] {$z$};

  \draw[edge] (n5a) -- (n6a) ;
  \draw[edge] (n5a) -- (n7a) ;
 
\end{tikzpicture}
}

\scalebox{.8}{

\begin{tikzpicture}
\tikzstyle{vertex}=[circle,draw = white, fill=black, minimum size = 7pt, inner sep=2pt]
\tikzstyle{vertex1}=[fill = white, draw = white]

\tikzstyle{edge}=[-,thick ]
\tikzstyle{elipse}=[-,thick ]
\tikzstyle{vertex2}=[circle,draw = black, fill=white, minimum size = 7pt, inner sep=2pt]

  \node[vertex2] (x) at (-0.75,0.5) {$3$};
  \node[vertex1] (n1) at (0,0) {} ;
  \node[vertex] (n2) at (0.5,0.5) {} ;
  \node[vertex1] (n3) at (0,1) {} ;
    
  \node[vertex] (n5) at (2.5,0.5) {} ;
  \node[vertex1] (n6) at (3,1) {} ;
  \node[vertex1] (n7) at (3,0) {} ;
  
  \draw[edge] (n1) -- (n2) ;
  \draw[edge] (n2) -- (n3) ;
  
  \path (n2) edge [bend left] node {} (n5) ;
  \path (n2) edge [bend right] node {} (n5) ;

  \draw[edge] (n6) -- (n5) ;
  \draw[edge] (n5) -- (n7) ;

  \node[vertex1] (n11) at (4, 0.5) {};
  \node[vertex1] (n12) at (5, 0.5) {};
  \draw[edge,->,very thick,>=stealth] (n11) -- (n12) ;
  
  \node[vertex1] (n1a) at (6,0) {} ;
  \node[vertex] (n2a) at (6.5,0.5) {} ;
  \node[vertex1] (n3a) at (6,1) {} ;

  \node[vertex] (n5a) at (8.5,0.5) {} ;
  \node[vertex1] (n6a) at (9,1) {} ;
  \node[vertex1] (n7a) at (9,0) {} ;

  \draw[edge] (n1a) -- (n2a) ;
  \draw[edge] (n2a) -- (n3a) ;
  
  \draw[edge] (n2a) -- (n5a);

  \draw[edge] (n5a) -- (n6a) ;
  \draw[edge] (n5a) -- (n7a) ;
 
\end{tikzpicture}
}

\scalebox{.8}{

\begin{tikzpicture}
\tikzstyle{vertex}=[circle,draw = white, fill=black, minimum size = 7pt, inner sep=2pt]
\tikzstyle{vertex1}=[fill = white, draw = white]

\tikzstyle{edge}=[-,thick ]
\tikzstyle{elipse}=[-,thick ]
\tikzstyle{vertex2}=[circle,draw = black, fill=white, minimum size = 7pt, inner sep=2pt]

  \node[vertex2] (x) at (-0.75,1.5) {$4$};

  \node[vertex1] (n1) at (0,0) {} ;
  \node[vertex] (n2) at (0.5,0.5) {} ;
  \node[vertex1] (n3) at (0,1) {} ;
  
  \node[vertex] (n4) at (1.5,1.25) {} ;
  
  \node[vertex] (n5) at (2.5,0.5) {} ;
  \node[vertex1] (n6) at (3,1) {} ;
  \node[vertex1] (n7) at (3,0) {} ;
  
  \node[vertex] (n8) at (1.5,2.5) {} ;
  \node[vertex1] (n9) at (1,3) {} ;
  \node[vertex1] (n10) at (2,3) {} ;
  
  \draw[edge] (n1) -- (n2) ;
  \draw[edge] (n2) -- (n3) ;
  \draw[edge] (n2) node[right = 2.5pt] {$u$} -- (n4) node[right = 2.5pt] {$x$};
  
  \draw[edge] (n4) -- (n5) ;
  \draw[edge] (n5) node[left = 2.5pt] {$v$} -- (n6) ;
  \draw[edge] (n5) -- (n7) ;

  \draw[edge] (n4) -- (n8) node[right = 2.5pt] {$w$} ;
  \draw[edge] (n8) -- (n9) ;
  \draw[edge] (n8) -- (n10) ;

  \node[vertex1] (n11) at (4, 1.75) {};
  \node[vertex1] (n12) at (5, 1.75) {};
  \draw[edge,->,very thick,>=stealth] (n11) -- (n12) ;
  
  \node[vertex1] (n1a) at (6,0) {} ;
  \node[vertex] (n2a) at (6.5,0.5) {} ;
  \node[vertex1] (n3a) at (6,1) {} ;
  
  
  \node[vertex] (n5a) at (8.5,0.5) {} ;
  \node[vertex1] (n6a) at (9,1) {} ;
  \node[vertex1] (n7a) at (9,0) {} ;
  
  \node[vertex] (n8a) at (7.5,2.5) {} ;
  \node[vertex1] (n9a) at (7,3) {} ;
  \node[vertex1] (n10a) at (8,3) {} ;
  
  \draw[edge] (n1a) -- (n2a) ;
  \draw[edge] (n2a) -- (n3a) ;
  \draw[edge] (n2a) node[below = 2.5pt] {$u$} -- (n5a) node[below = 2.5pt] {$v$};
  
  \draw[edge] (n5a) -- (n8a) ;
  \draw[edge] (n5a) -- (n6a) ;
  \draw[edge] (n5a) -- (n7a) ;

  \draw[edge] (n2a) -- (n8a) node[right = 2.5pt] {$w$} ;
  \draw[edge] (n8a) -- (n9a) ;
  \draw[edge] (n8a) -- (n10a) ;
 
\end{tikzpicture}
}

\scalebox{.8}{

\begin{tikzpicture}
\tikzstyle{vertex}=[circle,draw = white, fill=black, minimum size = 7pt, inner sep=2pt]
\tikzstyle{vertex1}=[fill = white, draw = white]

\tikzstyle{edge}=[-,thick ]
\tikzstyle{elipse}=[-,thick ]
\tikzstyle{vertex2}=[circle,draw = black, fill=white, minimum size = 7pt, inner sep=2pt]

  \node[vertex2] (x) at (-0.75,1.5) {$5$};

  \node[vertex1] (n1) at (0,0) {} ;
  \node[vertex] (n2) at (0.5,0.5) {} ;
  \node[vertex1] (n3) at (0,1) {} ;
  
  
  \node[vertex] (n5) at (2.5,0.5) {} ;
  \node[vertex1] (n6) at (3,1) {} ;
  \node[vertex1] (n7) at (3,0) {} ;
  
  \node[vertex] (n8) at (1.5,2.5) {} ;
  \node[vertex1] (n9) at (1,3) {} ;
  \node[vertex1] (n10) at (2,3) {} ;
  
  \draw[edge] (n1) -- (n2) ;
  \draw[edge] (n2) -- (n3) ;
  \draw[edge] (n2) node[below = 2.5pt] {$x$} -- (n5) node[below = 2.5pt] {$y$};
  
  \draw[edge] (n8) -- (n5) ;
  \draw[edge] (n5) -- (n6) ;
  \draw[edge] (n5) -- (n7) ;

  \draw[edge] (n2) -- (n8) node[right = 2.5pt] {$z$} ;
  \draw[edge] (n8) -- (n9) ;
  \draw[edge] (n8) -- (n10) ;

  \node[vertex1] (n11) at (4, 1.75) {};
  \node[vertex1] (n12) at (5, 1.75) {};
  \draw[edge,->,very thick,>=stealth] (n11) -- (n12) ;
  
  \node[vertex1] (n1a) at (6,0) {} ;
  \node[vertex] (n2a) at (6.5,0.5) {} ;
  \node[vertex1] (n3a) at (6,1) {} ;
  
  \node[vertex] (n4a) at (7.5,1.25) {} ;
  
  \node[vertex] (n5a) at (8.5,0.5) {} ;
  \node[vertex1] (n6a) at (9,1) {} ;
  \node[vertex1] (n7a) at (9,0) {} ;
  
  \node[vertex] (n8a) at (7.5,2.5) {} ;
  \node[vertex1] (n9a) at (7,3) {} ;
  \node[vertex1] (n10a) at (8,3) {} ;
  
  \draw[edge] (n1a) -- (n2a) ;
  \draw[edge] (n2a) -- (n3a) ;
  \draw[edge] (n2a) node[right = 2.5pt] {$x$} -- (n4a) ;
  
  \draw[edge] (n5a) node[left = 2.5pt] {$y$} -- (n4a) ;
  \draw[edge] (n5a) -- (n6a) ;
  \draw[edge] (n5a) -- (n7a) ;

  \draw[edge] (n8a) node[right = 2.5pt] {$z$} -- (n4a) node[right = 2.5pt] {$w$} ;
  \draw[edge] (n8a) -- (n9a) ;
  \draw[edge] (n8a) -- (n10a) ;
 
\end{tikzpicture}
}

\caption{Wye-Delta operations: 1. Degree-one reduction; 2. Series reduction; 3. Parallel reduction; 4. Wye-Delta transformation; 5. Delta-Wye transformation.}
~\label{fig:wyedelta}
\end{figure}
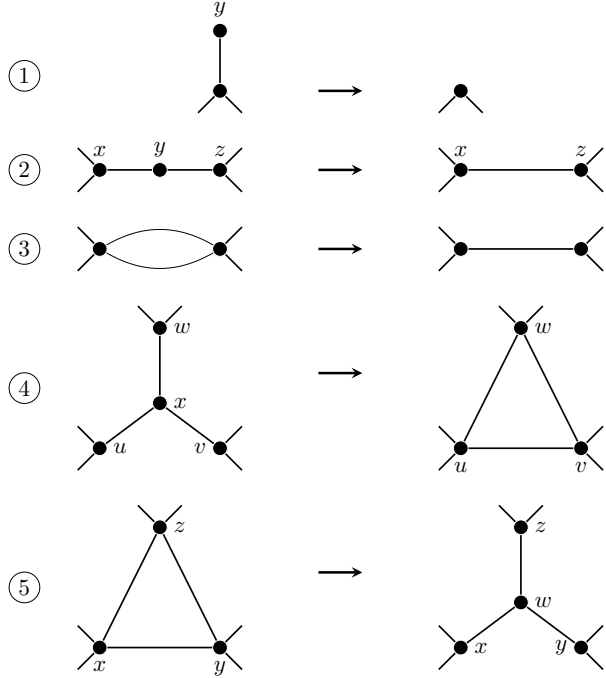

The following lemma (which follows from the above definitions) shows that the above rules preserve exactly all terminal minimum cuts.

\begin{lemma} \label{lemm: rules} Let $G$ be a $k$-terminal graph and $G'$ be a $k$-terminal graph obtained from $G$ by applying one of the rules $1-5$. Then $G'$ is a quality-$1$ cut sparsifier of $G$.
\end{lemma}

For our application, it will be useful to enrich the set of rules by introducing two new operations. These operations can be realized as series of the operations 1-5. (See Fig.~\ref{fig:edgedeletion} and \ref{fig:edgereplacement} for illustrations.)
\begin{enumerate}[resume]
	\item \emph{Edge deletion (with vertex $x$):} For a degree-three non-terminal with neighbours $u,v$, the edge $(u,v)$ can be deleted, if it exists. To achieve this, we use a Delta-Wye transformation followed by a series reduction.

	\item \emph{Edge replacement:} For a degree-four non-terminal vertex with neighbours $x,u,v,w$, if the edge $(x,u)$ exists, then it can be replaced by the edge $(v,w)$. To achieve this, we use a Delta-Wye transformation followed by a Wye-Delta transformation.
\end{enumerate}

\begin{figure}[H]
\begin{center}
\scalebox{.8}{
\begin{minipage}{.30\textwidth}
\centering

\begin{tikzpicture}
\tikzstyle{vertex}=[circle,draw = white, fill=black, minimum size = 7pt, inner sep=2pt]
\tikzstyle{vertex1}=[fill = white, draw = white]

\tikzstyle{edge}=[-,thick ]
\tikzstyle{elipse}=[-,thick ]
\tikzstyle{vertex2}=[circle,draw = black, fill=white, minimum size = 7pt, inner sep=2pt]

  \node[vertex2] (x) at (-0.75,1.5) {$6$};

  \node[vertex1] (n1) at (0,0) {} ;
  \node[vertex] (n2) at (0.5,0.5) {} ;
  \node[vertex1] (n3) at (0,1) {} ;
  
  \node[vertex] (n4) at (1.5,1.25) {} ;
  
  \node[vertex] (n5) at (2.5,0.5) {} ;
  \node[vertex1] (n6) at (3,1) {} ;
  \node[vertex1] (n7) at (3,0) {} ;
  
  \node[vertex] (n8) at (1.5,2.5) {} ;
  \node[vertex1] (n9) at (1,3) {} ;
  \node[vertex1] (n10) at (2,3) {} ;
  
  \draw[edge] (n1) -- (n2) node[below right] {$u$} ;
  \draw[edge] (n2) -- (n3) ;
  \draw[edge] (n2) -- (n4) node[right = 2.5pt] {$x$};
  
  \draw[edge] (n2) -- (n5); 
  
  \draw[edge] (n4) -- (n5) node[below left] {$v$} ;
  \draw[edge] (n5) -- (n6) ;
  \draw[edge] (n5) -- (n7) ;

  \draw[edge] (n4) -- (n8) ;
  \draw[edge] (n8) -- (n9) ;
  \draw[edge] (n8) -- (n10);
  
\end{tikzpicture}
\end{minipage} 
\hfill
\begin{minipage}{.25\textwidth}
\centering
\begin{tikzpicture}
\tikzstyle{vertex}=[circle,draw = white, fill=black, minimum size = 7pt, inner sep=2pt]
\tikzstyle{vertex1}=[fill = white, draw = white]

\tikzstyle{edge}=[-,thick ]
\tikzstyle{elipse}=[-,thick ]

  \node[vertex1] (n1) at (0,0) {} ;
  \node[vertex] (n2) at (0.5,0.5) {} ;
  \node[vertex1] (n3) at (0,1) {} ;
  
  \node[vertex] (n4) at (1.5,1.25) {} ;
  \node[vertex] (n11) at (1.5,0.625) {};
  
  \node[vertex] (n5) at (2.5,0.5) {} ;
  \node[vertex1] (n6) at (3,1) {} ;
  \node[vertex1] (n7) at (3,0) {} ;
  
  \node[vertex] (n8) at (1.5,2.5) {} ;
  \node[vertex1] (n9) at (1,3) {} ;
  \node[vertex1] (n10) at (2,3) {} ;
  
  \draw[edge] (n1) -- (n2) node[below right] {$u$} ;
  \draw[edge] (n2) -- (n3) ;
  \draw[edge] (n2) -- (n11) node[below = 2.5pt] {$w$} ;
 
  \draw[edge] (n4) node[right = 2.5pt] {$x$} -- (n11);
  
  \draw[edge] (n11) -- (n5) node[below left] {$v$} ;
  \draw[edge] (n5) -- (n6) ;
  \draw[edge] (n5) -- (n7) ;

  \draw[edge] (n4) -- (n8) ;
  \draw[edge] (n8) -- (n9) ;
  \draw[edge] (n8) -- (n10);
\end{tikzpicture}
\end{minipage}
\hfill
\begin{minipage}{.25\textwidth}
\centering
\begin{tikzpicture}
\tikzstyle{vertex}=[circle,draw = white, fill=black, minimum size = 7pt, inner sep=2pt]
\tikzstyle{vertex1}=[fill = white, draw = white]

\tikzstyle{edge}=[-,thick ]
\tikzstyle{elipse}=[-,thick ]

  \node[vertex1] (n1) at (0,0) {} ;
  \node[vertex] (n2) at (0.5,0.5) {} ;
  \node[vertex1] (n3) at (0,1) {} ;
  
  \node[vertex] (n4) at (1.5,1.25) {} ;
  
  \node[vertex] (n5) at (2.5,0.5) {} ;
  \node[vertex1] (n6) at (3,1) {} ;
  \node[vertex1] (n7) at (3,0) {} ;
  
  \node[vertex] (n8) at (1.5,2.5) {} ;
  \node[vertex1] (n9) at (1,3) {} ;
  \node[vertex1] (n10) at (2,3) {} ;
  
  \draw[edge] (n1) -- (n2) node[right = 2.5pt] {$u$};
  \draw[edge] (n2) -- (n3) ;
  \draw[edge] (n2) -- (n4) node[right = 2.5pt] {$w$}; 
  
  \draw[edge] (n4) -- (n5) node[left = 2.5pt] {$v$} ;
  \draw[edge] (n5) -- (n6) ;
  \draw[edge] (n5) -- (n7) ;

  \draw[edge] (n4) -- (n8) ;
  \draw[edge] (n8) -- (n9) ;
  \draw[edge] (n8) -- (n10);
  
\end{tikzpicture}
\end{minipage}
}
\end{center}
\caption{Edge deletion transformation. Edge capacities are omitted.}
\label{fig:edgedeletion}

\end{figure}
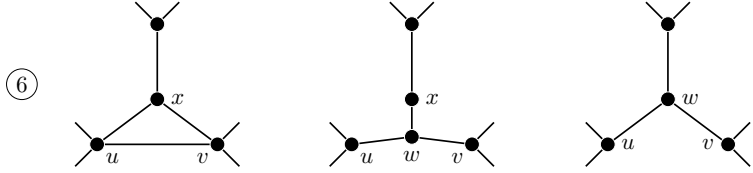

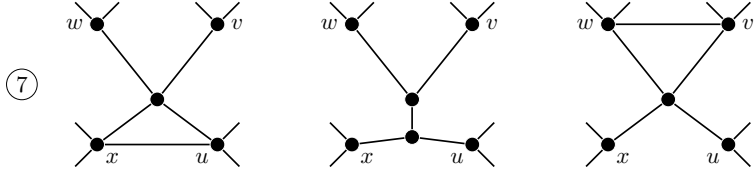
\begin{figure}[H]
\begin{center}
\scalebox{.8}{
\begin{minipage}{.30\textwidth}
\centering

\begin{tikzpicture}
\tikzstyle{vertex}=[circle,draw = white, fill=black, minimum size = 7pt, inner sep=2pt]
\tikzstyle{vertex1}=[fill = white, draw = white]

\tikzstyle{edge}=[-,thick ]
\tikzstyle{elipse}=[-,thick ]
\tikzstyle{vertex2}=[circle,draw = black, fill=white, minimum size = 7pt, inner sep=2pt]

  \node[vertex2] (x) at (-0.75,1.5) {$7$};

  \node[vertex1] (n1) at (0,0) {} ;
  \node[vertex] (n2) at (0.5,0.5) {} ;
  \node[vertex1] (n3) at (0,1) {} ;
  
  \node[vertex] (n4) at (1.5,1.25) {} ;
  
  \node[vertex] (n5) at (2.5,0.5) {} ;
  \node[vertex1] (n6) at (3,1) {} ;
  \node[vertex1] (n7) at (3,0) {} ;
  
  \node[vertex] (n8) at (0.5,2.5) {} ;
  \node[vertex1] (n9) at (0,3) {} ;
  \node[vertex1] (n10) at (1,3) {} ;
  
  \node[vertex] (n11) at (2.5,2.5) {} ;
  \node[vertex1] (n12) at (2,3) {} ;
  \node[vertex1] (n13) at (3,3) {} ;
  
  \draw[edge] (n1) -- (n2) node[below right] {$x$}  ;
  \draw[edge] (n2) -- (n3) ;
  \draw[edge] (n2) -- (n4) ;
  
  \draw[edge] (n2) -- (n5); 
  
  \draw[edge] (n4) -- (n5) node[below left] {$u$} ;
  \draw[edge] (n5) -- (n6) ;
  \draw[edge] (n5) -- (n7) ;

  \draw[edge] (n4) -- (n8) node[left = 2.5pt] {$w$};
  \draw[edge] (n8) -- (n9) ;
  \draw[edge] (n8) -- (n10);
  
  \draw[edge] (n4) -- (n11) node[right = 2.5pt] {$v$};
  \draw[edge] (n11) -- (n12) ;
  \draw[edge] (n11) -- (n13);
  
\end{tikzpicture}
\end{minipage} 
\hfill
\begin{minipage}{.25\textwidth}
\centering
\begin{tikzpicture}
\tikzstyle{vertex}=[circle,draw = white, fill=black, minimum size = 7pt, inner sep=2pt]
\tikzstyle{vertex1}=[fill = white, draw = white]

\tikzstyle{edge}=[-,thick ]
\tikzstyle{elipse}=[-,thick ]

  \node[vertex1] (n1) at (0,0) {} ;
  \node[vertex] (n2) at (0.5,0.5) {} ;
  \node[vertex1] (n3) at (0,1) {} ;
  
  \node[vertex] (n4) at (1.5,1.25) {} ;
  \node[vertex] (n11) at (1.5,0.625) {};
  
  \node[vertex] (n5) at (2.5,0.5) {} ;
  \node[vertex1] (n6) at (3,1) {} ;
  \node[vertex1] (n7) at (3,0) {} ;
  
  \node[vertex] (n8) at (0.5,2.5) {} ;
  \node[vertex1] (n9) at (0,3) {} ;
  \node[vertex1] (n10) at (1,3) {} ;
  
  \node[vertex] (n12) at (2.5,2.5) {} ;
  \node[vertex1] (n13) at (2,3) {} ;
  \node[vertex1] (n14) at (3,3) {} ;
  
  \draw[edge] (n1) -- (n2) node[below right] {$x$} ;
  \draw[edge] (n2) -- (n3) ;
  \draw[edge] (n2) -- (n11);
 
  \draw[edge] (n4) -- (n11);
  
  \draw[edge] (n11) -- (n5) node[below left] {$u$};
  \draw[edge] (n5) -- (n6) ;
  \draw[edge] (n5) -- (n7) ;

  \draw[edge] (n4) -- (n8) node[left = 2.5pt] {$w$};
  \draw[edge] (n8) -- (n9) ;
  \draw[edge] (n8) -- (n10);
  
  \draw[edge] (n4) -- (n12) node[right = 2.5pt] {$v$};
  \draw[edge] (n12) -- (n13) ;
  \draw[edge] (n12) -- (n14);
\end{tikzpicture}
\end{minipage}
\hfill
\begin{minipage}{.25\textwidth}
\centering
\begin{tikzpicture}
\tikzstyle{vertex}=[circle,draw = white, fill=black, minimum size = 7pt, inner sep=2pt]
\tikzstyle{vertex1}=[fill = white, draw = white]

\tikzstyle{edge}=[-,thick ]
\tikzstyle{elipse}=[-,thick ]

  \node[vertex1] (n1) at (0,0) {} ;
  \node[vertex] (n2) at (0.5,0.5) {} ;
  \node[vertex1] (n3) at (0,1) {} ;
  
  \node[vertex] (n4) at (1.5,1.25) {} ;
  
  \node[vertex] (n5) at (2.5,0.5) {} ;
  \node[vertex1] (n6) at (3,1) {} ;
  \node[vertex1] (n7) at (3,0) {} ;
  
  \node[vertex] (n8) at (0.5,2.5) {} ;
  \node[vertex1] (n9) at (0,3) {} ;
  \node[vertex1] (n10) at (1,3) {} ;
  
  \node[vertex] (n11) at (2.5,2.5) {} ;
  \node[vertex1] (n12) at (2,3) {} ;
  \node[vertex1] (n13) at (3,3) {} ;
  
  \draw[edge] (n1) -- (n2) node[below right] {$x$};
  \draw[edge] (n2) -- (n3) ;
  \draw[edge] (n2) -- (n4) ;
  
  \draw[edge] (n8) -- (n11); 
  
  \draw[edge] (n4) -- (n5) node[below left] {$u$};
  \draw[edge] (n5) -- (n6) ;
  \draw[edge] (n5) -- (n7) ;

  \draw[edge] (n4) -- (n8) node[left = 2.5pt] {$w$};
  \draw[edge] (n8) -- (n9) ;
  \draw[edge] (n8) -- (n10);
  
  \draw[edge] (n4) -- (n11) node[right = 2.5pt] {$v$};
  \draw[edge] (n11) -- (n12) ;
  \draw[edge] (n11) -- (n13);
  
\end{tikzpicture}
\end{minipage}
}
\end{center}
\caption{Edge replacement transformation. Edge capacities are omitted.}
\label{fig:edgereplacement}
\end{figure}

A $k$-terminal graph $G$ is \emph{Wye-Delta} reducible to another $k$-terminal graph $H$, if $G$ is reduced to $H$ by repeatedly applying one of the operations 1-7. 

\begin{lemma} \label{lemm: cutReducability} Let $G$ and $H$ be $k$-terminal graphs. Moreover, let $G$ be Wye-Delta reducible to $H$. Then $H$ is a quality-$1$ cut sparsifier of $G$.
\end{lemma}
\begin{proof}
	Observe that the rules 1-7 do not affect any terminal vertex and each rule preserves exactly all terminal minimum cuts by Lemma \ref{lemm: rules}. An induction on the number of rules needed to reduce $G$ to $H$ proves the claim.
\end{proof}
\paragraph*{Grid Graphs.} 
A \emph{grid} graph is a graph with $n \times n$ vertices $\{(u,v) : u,v = 1,\ldots, n\}$, where $(u,v)$ and $(u',v')$ are adjacent if $|u' -u| + |v'-v| = 1$. For $k < n$, a \emph{half-grid} graph with $k$ terminals is a graph $T^{n}_k = (V,E)$ with $K \subset V$ and $n(n+1)/2$ vertices $\{(i,j) : i \leq j \text{ and } i,j=1,\ldots,n  \}$, where $(i,j)$ and $(i',j')$ are connected by an edge if $|i'-i| + |j'-j| = 1$, and additional diagonal edges between $(i,i)$ and $(i+1,i+1)$ for $i = 1,\ldots, n-1$. Moreover, each terminal vertex in $T^{n}_k$ must be one of its diagonal vertices, i.e., every $x \in K$ is of the form $(m,m)$ for some $m \in \{1,\ldots,n\}$. Let $\hat{T}^{n}_{k}$ be the same graph as $T^{n}_k$ but excluding the diagonal edges. 

\paragraph*{Graph Embeddings.}\label{sec: graphEmbeddings} 
Throughout this paper, we will be dealing with the embedding of a planar graph into a square \emph{grid} graph. One way of drawing graphs in the plane are \emph{orthogonal grid-embeddings}~\cite{Valiant81}. In such a setting, the vertices correspond to distinct points and edges consist of alternating sequences of vertical and horizontal segments. Equivalently, one can view this as drawing our input graph as a subgraph of some grid. Formally, a \emph{node-embedding} $\rho$ of $G_1=(V_1,E_1)$ into $G_2=(V_2,E_2)$ is an injective mapping that maps $V_1$ into $V_2$, and $E_1$ into paths in $G_2$, i.e., $(u,v)$ maps to a path from $\rho(u)$ to $\rho(v)$, such that every pair of paths that correspond to two different edges in $G_1$ is vertex-disjoint (except possibly at the endpoints). If $G_2$ is a planar graph, then $\rho(G_1)$ and $G_1$ are also planar. Thus, if $G_1$ and $G_2$ are planar we then refer to $\rho$ as an \emph{orthogonal embedding}. Moreover, given a planar graph $G_1$ drawn in the plane, the embedding $\rho$ is called \emph{region-preserving} if $\rho(G_1)$ and $G_1$ have the same planar topological embedding.


Let $G_1$ be a $k$-terminal graph. Since the embedding does not affect the vertices of $G_1$, the terminals of $G_1$ are also terminals in $\rho(G_1)$. Although the embedding does not consider capacity of the edges in $G_1$, we can still guarantee that such an embedding preserves all terminal minimum cuts, for which we make use of the following operation:
\begin{enumerate}
	\item \emph{Edge subdivision: } Let $(u,v)$ be an edge of capacity $c(u,v)$. Delete $(u,v)$, introduce a new vertex $w$ and add edges $(u,w)$ and $(w,v)$, each of capacity $c(u,v)$. 
\end{enumerate}

The following lemma shows that a node-embedding is a cut preserving mapping.
\begin{lemma} \label{lemm: cutembedding}
	Let $\rho$ be a node-embedding and let $G_1$ and $\rho(G_1)$ be $k$-terminal graphs defined as above. Then $\rho(G_1)$ preserves exactly all terminal minimum cuts of $G$.
\end{lemma}
\begin{proof}
	We can view each path obtained from the embedding as taking the edge corresponding to the path endpoints in $G_1$ and performing edge subdivisions finitely many times. We claim that such subdivisions preserve all terminal cuts. 
	
	Indeed, let us consider a single edge subdivision for $(u,v)$ (the general claim then follows by induction on the number of edge subdivisions). Fix $S \subset K$ and consider some $S$-separating minimum cut $(U,V \setminus U)$ in $G_1$ cutting $(u,v)$. Then, in the transformed graph $\rho(G_1)$, we can simply cut either the edge $(u,w)$ or $(w,v)$. Since by construction, the new edge has the same capacity as the subdivided edge, we get that $\capacity_{\rho(G_1)}(\delta(U)) = \capacity_{G_1}(\delta(U))$, and in particular $\mincut_{\rho(G_1)}(S,K \setminus S) \leq \mincut_{G_1}(S,K \setminus S)$. 
	
	Furthermore, since $G_1$ is obtained by contracting two edges of the same capacity of $\rho(G_1)$, for any $S$-separating minimum cut $(U, V \setminus U)$ in $\rho(G_1)$, we have $\capacity_{\rho(G_1)}(\delta(U)) \geq \capacity_{G_1}(\delta(U))$, and in particular $\mincut_{\rho(G_1)}(S,K \setminus S) \geq \mincut_{G_1}(S,K \setminus S)$. Combining the above gives the lemma. 
\end{proof}
%
%
\subsection{Our Construction}
In this section  we construct our exact cut sparsifier and prove that any planar $k$-terminal graph with all terminals lying on the same face admits a cut sparsifier of size $O(k^2)$ that is also planar.
\subsubsection{Embedding into Grids} \label{EmbeddingGrid}

It is well-known that one can obtain an orthogonal embedding of a planar graph with maximum-degree at most three into a grid (see Valiant~\cite{Valiant81}). However, our input planar graph can have arbitrarily large maximum degree. In order to be able to make use of such an embedding, we need to first reduce our input graph to a bounded-degree graph while preserving planarity and all terminal minimum cuts. We achieve this by making use of a \emph{vertex splitting} technique, which we describe below. 

Given a $k$-terminal planar graph $G'=(V',E',c')$ with $K \subset V'$ lying on the outer face, vertex splitting produces a $k$-terminal planar graph $G=(V,E,c)$ with $K \subset V$ such that the maximum degree of $G$ is at most three. Specifically, for each vertex $v$ of degree $d>3$ with neighboring vertices $u_1,\ldots,u_{d}$, we delete $v$ and introduce new vertices $v_1, \ldots, v_d$ along with edges $\{(v_i,v_{i+1}) : i=1,\ldots,d-1\}$, each of capacity $C+1$, where $C=\sum_{e \in E'} c'(e)$. Further, we replace the edges $\{(u_i,v) : i = 1,\ldots,d\}$ with $\{(u_i,v_i) : i = 1,\ldots,d\}$, each of corresponding capacity. If $v$ is a terminal vertex, we set one of the $v_i$'s to be a terminal vertex. It follows that the resulting graph $G$ is planar and terminals can be still embedded on the outer face. Note that while the degree of every vertex $v_i$ is at most $3$, the degree of any other vertex is not affected. 

\begin{claim} \label{claim: embedding}
Let $G'$ and $G$ be $k$-terminal graphs defined as above. Then $G$ preserves exactly all minimum terminal cuts of $G'$, i.e., $G$ is a quality-$1$ cut sparsifier of $G'$.
\end{claim} 
\begin{proof}
	It suffices to prove the case where $G$ is obtained from $G'$ by a single vertex splitting. Then the claim follows by induction on the number of vertex splittings required to transform $G'$ to $G$.
	
	Let $S \subset K$ and $(U, V \setminus U)$ be an $S$-separating cut in $G$ of size $\mincut_G(S,K \setminus S)$. Suppose towards contradiction that $\delta(U)$ contains an edge of the form $(v_j,v_{j+1})$, for some $j$, which in turn gives that $\capacity(\delta(U)) \geq C + 1$. Then we can move all the points $v_i$ to one of the sides of the cut $(U, V \setminus S)$ and obtain a new $S$-separating cut in $G$ of cost at most $C$, contradicting the fact that $(U, V \setminus U)$ is a minimum terminal cut. Hence, it follows that $\delta(U)$ uses either edges that are in both $G$ and $G'$ or edges of the form $(u_i,v_i)$, which by construction have the same capacity as the edges $(u_i,v)$ in $G'$. Thus, an $S$-separating minimum cut in $G$ corresponds to an $S$-separating minimum cut in $G'$ of the same cost. Since $S$ is chosen arbitrarily, the claim follows.
\end{proof}
%

Let $G=(V,E)$ be a $k$-terminal graph obtained by vertex splitting of all vertices of degree larger than $3$ of $G'=(V',E')$. Further, let $n' = |V'|$, $m' = |E'|$, $n = |V|$ and $m = |E|$. Then it is easy to show that $n \leq 2 m'$ and $m \leq m' + n \leq 3m'$. Since $G'$ is planar, we have that $n = O(n')$ and $m = O(n')$. Thus, by just a linear blow-up on the size of vertex and edge sets, we may assume w.l.o.g. that our input graph is a planar graph of degree at most three. 

Valiant~\cite{Valiant81} and Tamassia et al.~\cite{TamassiaT89} showed that a $k$-terminal planar graph $G$ with $n$ vertices and degree at most three admits an orthogonal region-preserving embedding into some square grid of size $O(n) \times O(n)$. By Lemma \ref{lemm: cutembedding}, we know that the resulting graph exactly preserves all terminal minimum cuts of $G$. We remark that since the embedding is region-preserving, the outer face of the input graph is embedded to the outer face of the grid. Therefore, all terminals in the embedded graph lie on the outer face of the grid. Performing appropriate edge subdivisions, we can make all the terminals lie on the boundary of some possibly larger grid. Further, we can add dummy non-terminals and zero edge capacities to transform our graph into a full-grid $H$. We observe that the latter does not affect any terminal min-cut. The above leads to the following:
\begin{lemma} \label{lemm: embeddingGrids}
Given a $k$-terminal planar graph $G$, where all terminals lie on the outer face, there exists a $k$-terminal grid graph $H$, where all terminals lie on the boundary such that $H$ preserves exactly all terminal minimum cuts of $G$. The resulting graph has $O(n^2)$  vertices and edges.
\end{lemma}

\subsubsection{Embedding Grids into Half-Grids}

Next, we show how to embed square grids into half-grid graphs (see Section \ref{sec: preli}), which will facilitate the application of Wye-Delta transformations. The existence of such an embedding was claimed in the thesis of Gitler~\cite{Gitler91}, but no details on its construction were given. 

Let $G$ be a $k$-terminal square grid on $n \times n$ vertices where terminals lie on the boundary of the grid. We obtain the following:
\begin{figure}[t] 
\begin{center}
\includegraphics[scale=1]{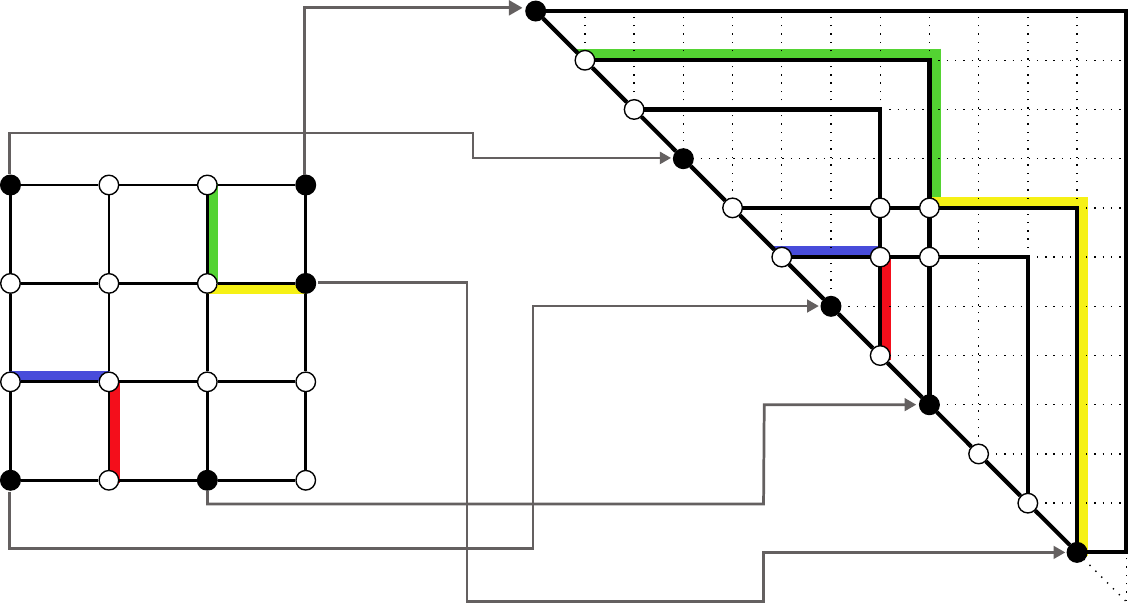}
\end{center}
\caption{Embedding grid into half-grid. Black vertices represent terminals while white vertices represent non-terminals. The counter-clockwise ordering starts at the top right terminal. Coloured edges and paths correspond to the mapping of the respective edges: blue for edges $((i,1),(i,2))$, red for edges $((n-1,j),(n,j))$, green for edges $((1,j),(2,j))$ and yellow for edges $((i,n-1), (i,n))$, where $i,j = 2,\ldots,n-1$. }
\label{fig: halfgridConst}
\end{figure}

\begin{lemma} \label{lemm: embeddingHalfGrid}
There exists a node embedding of the grid $G$ into $T^{\ell}_{k}$, where $\ell  = 4n-3$. 
\end{lemma}
\begin{proof}
Our construction works as follows (See Fig.~\ref{fig: HalfGridRed} for an example). We first fix an ordering on the vertices lying on the boundary of the grid in the order induced by the grid. Then we embed each vertex according to that order into the diagonal vertices of the half-grid, along with the edges that form the boundary of the grid. The sub-grid obtained by removing all boundary vertices is embedded appropriately into the upper-part of the half-grid. Finally, we show how to embed edges between the boundary and the sub-grid vertices and argue that such an embedding is indeed vertex-disjoint for any pair of paths.

\input{figure_halfgridReduction}

We start with the embedding of the vertices of $G$. Let us first consider the boundary vertices. The ordering imposed on these vertices can be viewed as starting with the upper-right vertex $(1,n)$ and visiting the rest of vertices in a counter-clockwise direction until reaching the vertex $(2,n)$. We map the vertices on the boundary as follows.
\begin{enumerate}
\setlength\itemsep{0.1em}
\item The vertex $(1,j)$ is mapped to the vertex $(n-j+1,n-j+1)$ for $j=2,\ldots,n$, 
\item The vertex $(i,1)$ is mapped to the vertex $(n+i-1,n+i-1)$ for $i=1,\ldots,n-1$,
\item The vertex $(n,j)$ is mapped to the vertex $(2n+j-2,2n+j-2)$ for $j=1,\ldots,n-1$,
\item The vertex $(i,n)$ is mapped to the vertex $(4n-i-2,4n-i-2)$  for $i=2,\ldots,n$.
\end{enumerate} 
Now we consider the vertices that belong to the induced sub-grid $S$ of $G$ of size $(n-2)^2$ when removing the boundary vertices of our input grid. We map the vertex $(i,j)$ to the vertex $(n+i-1, 2n+j-2)$ for $i,j=2,\ldots,n-1$. In other words, for every vertex of $S$ we make a vertical shift by $n-1$ units and an horizontal shift by $2n-2$ units. By construction, it is not hard to check that every vertex of $G$ is mapped to a different vertex of $T^{\ell}_{k}$ and all terminal vertices lie on the diagonal of $T^{\ell}_{k}$. 

We continue with the embedding of the edges of $G$. First, every edge between two boundary vertices in $G$ is embedded to the edge between the corresponding mapped diagonal vertices of $T^{\ell}_{k}$, except the edge between $(1,n)$ and $(2,n)$. For this edge, we define an edge embedding between the corresponding vertices $(1,1)$ and $(4n-4,4n-4)$ of $T^{\ell}_{k}$ by using the path:
\[ (1,1) \rightarrow (1,2) \rightarrow \ldots \rightarrow (1,4n-3) \rightarrow (2,4n-3) \rightarrow \ldots \rightarrow (4n-4,4n-3) \rightarrow (4n-4,4n-4).\]

Next, every edge of the sub-grid $S$ is embedded in to the edge connecting the mapped endpoints of that edge in $T^{\ell}_{k}$. In other words, if $(i,j)$ and $(i',j')$ were connected by an edge $e$ in $S$, then $(n+i-1, 2n+j-2)$ and $(n+i'-1, 2n+j'-2)$ are connected by an edge $e'$ in $T^{\ell}_{k}$ and $e$ is mapped to $e'$. Finally, the only edges that remain are those connecting a boundary vertex of $G$ with a boundary vertex of $S$. We distinguish four cases depending on the edge position. 
\begin{enumerate}
\setlength\itemsep{0.1em}
\item The edge $((i,2), (i,1))$ is mapped to the horizontal path given by:
\[
	(n+i-1,2n) \rightarrow (n+i-1, 2n-1) \rightarrow \ldots \rightarrow (n+i-1,n+i-1) \text{ for } i = 2,\ldots,n-1.
\]
\item The edge $((n-1,j), (n,j))$ is mapped to the vertical path given by:
\[
	(2n-2,2n+j-2) \rightarrow (2n-1, 2n+j-2) \rightarrow \ldots \rightarrow (2n+j-2,2n+j-2) \text{ for } j = 2,\ldots,n-1.
\]
\item The edge $((2,j), (1,j))$ is mapped to the $L$-shaped path:
\begin{align*}
	(n+1,2n+&j-2)  \rightarrow (n, 2n+j-2) \rightarrow \ldots \rightarrow (n-j+1,2n+j-2) \\
	&  \rightarrow (n-j+1,2n+j-3)  \rightarrow \ldots \rightarrow (n-j+1,n-j+1) \text{ for } j = 2,\ldots,n-1.
\end{align*}

\item The edge $((i,n-1), (i,n))$ is mapped to the $L$-shaped path:
\begin{align*}
	(n+i-1,3n&-3)  \rightarrow (n+i-1, 3n-2) \rightarrow \ldots \rightarrow (n+i-1,4n-i-2) \\
	&  \rightarrow (n+i,4n-i-2)  \rightarrow \ldots \rightarrow (4n-i-2,4n-i-2) \text{ for } i = 2,\ldots,n-1.
\end{align*}
\end{enumerate} 
By construction, it follows that the paths in our edge embedding are vertex disjoint.
\end{proof}

\subsubsection{Reducing Half-Grids and Bringing the Piece Together}
We now review the construction of Gitler~\cite{Gitler91}, which shows how to reduce half-grids to much smaller half-grids (excluding diagonal edges) whose size depends only on $k$. For the sake of completeness, we provide a full proof here. Recall that $\hat{T}^{n}_{k}$ is the graph $T^{n}_k$ without the diagonal edges. 

\begin{lemma}[\cite{Gitler91}] \label{lemm: Gitler}
For any positive $k,n$ with $k<n$, $T^{n}_k$ is Wye-Delta reducible to $\hat{T}^{k}_{k}$.
\end{lemma}
\begin{proof}
	For sake of simplicity, we assume w.l.o.g that the four vertices $(1,1)$, $(2,2)$, $(n-1,n-1)$ and $(n,n)$ are terminals\footnote{If they are not terminals, we can simply define them as terminals, thus increasing the number of terminals to $k+4 = O(k)$.}. Furthermore, we say that two terminals $(i,i)$ and $(j,j)$ are \textit{adjacent} iff $i<j$ and there is no terminal $(\ell,\ell)$ such that $i < \ell < j$. 
	
	We next describe the reduction procedure. Also see Fig.~\ref{fig: HalfGridRed} for an example. The reduction procedure starts by removing the diagonal edges of $T^{n}_k$, thus producing the graph $\hat{T}^{n}_k$ . Specifically, the two edges $((1,1),(2,2))$ and $((n-1,n-1), (n,n))$ are removed using an edge deletion operation. For each remaining diagonal edge of the form $((i,i), (i+1,i+1))$, $i=2,\ldots,n-2$ we repeatedly apply an edge replacement operation until the edge is incident to a boundary vertex $(1,j)$ or $(j,n)$ of the grid, where an edge deletion operation with one of the neighbours of $(1,j)$ resp. $(j,n)$ as vertex $x$ is applied.
	
	Now, we know that all non-terminals of the form $(i,i)$ are degree-two vertices, thus a series reduction is applied on each of them. This produces new diagonal edges, which are effectively reduced by the above procedure. We keep removing the newly-created degree-two non-terminal vertices and the newly-created edges until no further removals are possible. At this point, the only degree-$2$ vertices are terminal vertices. 
	
	The resulting graph has a staircase structure, where for every pair of adjacent terminals $(i,i)$ and $(j,j)$, there is a non-terminal $(i,j)$ of degree three or four, namely, the intersection vertex, and a (possibly empty) sequence of degree-three non-terminals that lie on the boundary path from $(i,i)$ to $(j,j)$. For $k = i+1,\ldots,j-1$, let $(i,k)$ and $(k,j)$ be the degree-three non-terminals lying on the row and the column subpath, respectively. Additionally, for $k = i+1,\ldots,j-1$, let $C^{i}_{k}=\{(i',k) : i'=i,\ldots,1\}$, resp. $R^{j}_{k}=\{(k,j') : j'=j,\ldots,n\}$ be the vertices sharing the same column, resp. row with $(i,k)$, resp. $(k,j)$. We next show that the vertices belonging to $C^{i}_{k}$ and $R^{j}_{k}$ can be removed. 
	
	The removal process works as follows. For $k=i+1,\ldots,j-1$, we start by choosing a degree $3$ vertex $(i,k)$ and its corresponding column $C^{i}_{k}$. Then we apply a Wye-Delta transformation on $(i,k)$, thus creating two new diagonal edges. Similarly as above, we remove such edges by repeatedly applying an edge replacement operation until they have been pushed to the boundary of the grid, where an edge deletion operation is applied. In the resulting graph, the vertex $(i-1,k) \in C^{i}_{k}$ is now a degree-three non-terminal. We apply the same procedure to this vertex. Applying such a procedure to all remaining vertices of $C^{i}_{k}$, we eliminate a column of the grid. Symmetrically, the same process applies to the case when we want to remove the row $R^{j}_{k}$ corresponding to the vertex $(k,j)$.
	
	Applying the above removal process for every adjacent terminal pair and the corresponding degree-three non-terminals, we end up with the graph $\hat{T}^{k}_{k}$, where every diagonal vertex is a terminal. By definition, it follows that $\hat{T}^{k}_{k}$ has at most $O(k^2)$ vertices.
\end{proof}

Combining the above reductions leads to the following theorem:
\begin{theorem} \label{thm: CutSparsifier}
Let $G$ be a $k$-terminal planar graph where all terminals lie on the outer face. Then $G$ admits a quality-$1$ cut sparsifier of size $O(k^{2})$, which is also a planar graph.
\end{theorem}
\begin{proof}
Let $n$ denote the number of vertices in $G$. First, we apply Lemma \ref{lemm: embeddingGrids} on $G$ to obtain a grid graph $H$ with $O(n^{2})$ vertices, which preserves exactly all terminal minimum cuts of $G$. We then apply Lemma \ref{lemm: embeddingHalfGrid} on $H$ to obtain a node embedding $\rho$ into the half-grid $T^{\ell}_{k}$, where $\ell = 4n-3$. By Lemma \ref{lemm: cutembedding}, $\rho(H)$ preserves exactly all terminal minimum cuts of $H$. We can further extend $\rho(H)$ to the full half-grid $T^{\ell}_{k}$, if dummy non-terminals and zero edge capacities are added. Finally, we apply Lemma \ref{lemm: Gitler} on $T^{\ell}_{k}$ to obtain a Wye-Delta reduction to the reduced half-grid graph $\hat{T}^{k}_{k}$. It follows by Lemma \ref{lemm: cutReducability} that $\hat{T}^{k}_{k}$ is a quality-$1$ cut sparsifier of $T^{\ell}_{k}$, where the size guarantee is immediate from the definition of $\hat{T}^{k}_{k}$.
\end{proof}

\section{Extensions to Planar Flow and Distance Sparsifiers} \label{sec: UpperFlow}
In this section we show how to extend our result for cut sparsifiers to flow and distance sparsifiers. 
\subsection{An Upper Bound for Flow Sparsifiers}
We first review the notion of Flow Sparsifiers. Let $d$ be a demand function over terminal pairs in $G$ such that $d(x,x')=d(x',x)$ and $d(x,x)=0$ for all $x,x' \in K$. We denote by $P_{xx'}$ the set of all paths between vertices $x$ and $x'$, for all $x,x' \in K$. Further, let $P_{e}$ be the set of all paths using edge $e$, for all $e \in E$ . A \emph{concurrent} (\textit{multi-commodity}) flow $f$ of \emph{throughput} $\lambda$ is a function over terminal paths in $G$ such that (1) $\sum_{p \in P_{xx'}} f(p) \geq \lambda d(x,x')$, for all distinct terminal pairs $x,x' \in K$ and (2) $\sum_{p \in P_e} f(p) \leq c(e)$, for all $e \in E$. We let $\lambda_G(d)$ denote the \emph{throughput of the concurrent flow} in $G$ that attains the largest throughput and we call a flow achieving this throughput the \emph{maximum concurrent flow}. A graph $H = (V_H, E_H, c_H)$, $K \subset V_H$ is a \emph{quality-$1$} (\emph{vertex}) \emph{flow sparsifier} of $G$ with $q \geq 1$ if for every demand function
$d$, $\lambda_G(d) \leq \lambda_H(d) \leq q \cdot \lambda_H(d).$

Next we show that given a $k$-terminal planar graph, where all terminals lie on the outer face, one can construct a quality-$1$ flow sparsifier of size $O(k^{2})$. Our result follows from combining the observation of Andoni et al.~\cite{andoni} for constructing flow-sparsifiers using flow/cut gaps and the flow/cut gap result of Okamura and Seymour~\cite{OkamuraS81}.

Given a $k$-terminal graph and a demand function $d$, recall that $\lambda_G(d)$ is the maximum fraction of $d$ that can be routed in $G$. We define the \emph{sparsity} of a cut $(U, V \setminus U)$ to be
\[
\Phi_G(U,d) := \frac{\capacity(\delta(U))}{\sum_{i,j: |\{i,j\} \cap U|=1}d_{ij}}
\]
and the \emph{sparsest cut} as $\Phi_G(d) := \min_{U \subset V} \Phi_G(U,d)$. Then the \emph{flow-cut} gap is given by
\[
\gamma(G) := \max \{\Phi_G(d) / \lambda_G(d) : d \in \mathbb{R}^{\binom{k}{2}}_{+}\}.
\]

We will make use of the following theorem:
\begin{theorem}[\cite{andoni}] \label{thm: FlowCutGap}
	Given a $k$-terminal graph $G$ with terminals $K$, let $G'$ be a quality-$\beta$ cut sparsifier for $G$ with $\beta \geq 1$. Then for every demand function $d \in \mathbb{R}^{\binom{k}{2}}_{+}$,
	\[
	\frac{1}{\gamma(G')} \leq \frac{\lambda_{G'}(d)}{\lambda_G(d)} \leq \beta \cdot \gamma(G).
	\]
	Therefore, the graph $G'$ with edge capacities scaled up by $\gamma(G')$ is a quality-$\beta \cdot \gamma(G) \cdot \gamma(G')$ flow sparsifier of size $|V(G')|$ for $G$.
\end{theorem}
This leads to the following corollary.
\begin{corollary} \label{cor: FlowSparsifiers}
	Let $G$ be a $k$-terminal planar graph where all terminals lie on the outer face. Then $G$ admits a quality-$1$ flow sparsifier of size $O(k^{2})$.
\end{corollary}
\begin{proof}
	Given a $k$-terminal planar graph where all terminals lie on the outer face, Theorem \ref{thm: CutSparsifier} shows how to construct a cut sparsifier $G'$ with quality $\beta = 1$ and size $O(k^2)$, which is also a planar graph with all the $k$ terminals lying on the outer face. Okamura and Seymour~\cite{OkamuraS81} showed that for every $k$-terminal planar graph $G$ with terminals lying on the outer face the flow-cut gap is $1$. This implies that $\gamma(G) = 1$ and $\gamma(G') = 1$. Invoking Theorem \ref{thm: FlowCutGap} we get that $G'$ is a quality-$1$ flow sparsifier of size $O(k^2)$ for $G$.
\end{proof}

\subsection{An Upper Bound for Distance Sparsifiers} \label{sec: UpperDist}
We first review the notion of Vertex Distance Sparsifiers. Let $G=(V,E,\ell)$ with $K \subset V$ be a $k$-terminal graph, where we replace the capacity function $c$ with a length function $\ell : E \rightarrow \mathbb{R}_{\geq 0}$. For a terminal pair $(x,x') \in K$, let $d_G(x,x')$ denote the shortest path with respect to the edge lengths $\ell$ in $G$. A graph $H=(V',E',\ell')$ is a \emph{quality-$q$} (\emph{vertex}) \emph{distance sparsifier} of $G$ with $q \geq 1$ if for any $x,x' \in K$, $d_G(x,x') \leq d_H(x,x') \leq q \cdot d_G(x,x')$.

Next we argue that a symmetric approach applies to the construction of vertex sparsifiers that preserve distances. Concretely, we prove that given a $k$-terminal planar graph, where all terminals lie on the outer face, one can construct a quality-$1$ distance sparsifier of size $O(k^{2})$, which is also a planar graph. It is not hard to see that almost all arguments that we used about cut sparsifiers go through, except some adaptations regarding edge lengths in the Wye-Delta rules, edge subdivision operation and vertex splitting operation.

We start adapting the Wye-Delta operations. 

\begin{enumerate}
	\setlength\itemsep{0.1em}
	\item \emph{Degree-one reduction:} Delete a degree-one non-terminal and its incident edge.
	\item \emph{Series reduction:} Delete a degree-two non-terminal $y$ and its incident edges $(x,y)$ and $(y,z)$, and add a new edge $(x,z)$ of length $\ell(x,y) + \ell(y,z)$.
	\item \emph{Parallel reduction:} Replace all parallel edges by a single edge whose length is the minimum over all lengths of parallel edges.
	\item \emph{Wye-Delta transformation:} Let $x$ be a degree-three non-terminal with neighbours $\delta(x) = \{u,v,w\}$. 
	Delete $x$ (along with all its incident edges) and add edges $(u,v),(v,w)$ and $(w,u)$ with lengths $\ell(u,x) + \ell(v,x)$, $\ell(v,x) + \ell(w,x)$ and $\ell(w,x) + \ell(u,x)$, respectively. 
	\item \emph{Delta-Wye transformation:} Let $x$, $y$ and $z$ be the vertices of the triangle connecting them. Assume w.l.o.g.\footnote{Suppose there exists a triangle edge $(x,y)$ with $\ell(x,y) > \ell(x,z) + \ell(y,z)$, where $z$ is the other triangle vertex. Then we can simply set $\ell(x,y) = \ell(x,z) + \ell(y,z)$, since any shortest path between terminal pairs would use the edges $(x,z)$ and $(y,z)$ instead of the edge $(x,y)$.} that for any triangle edge $(x,y)$, $\ell(x,y) \leq \ell(x,z) + \ell(y,z)$, where $z$ is the other triangle vertex. Delete the edges of the triangle, introduce a new vertex $w$ and add new edges $(w,x)$, $(w,y)$ and $(w,z)$ with edge lengths $(\ell(x,y) + \ell(x,z) - \ell(y,z))/2,$ $(\ell(x,z) + \ell(y,z) - \ell(x,u))/2$ and $(\ell(x,y) + \ell(y,z) - \ell(x,z))/2$, respectively.
\end{enumerate}
The following lemma shows that the above rules preserve exactly all shortest path distances between terminal pairs. 

\begin{lemma} Let $G$ be a $k$-terminal graph and $G'$ be a $k$-terminal graph obtained from $G$ by applying one of the rules 1-5. Then $G'$ is a quality-$1$ distance sparsifier of $G$.
\end{lemma}
We remark that there is no need to re-define the Edge deletion and replacement operations, since they are just a combination of the above rules. An analogue of Lemma \ref{lemm: cutReducability} can also be shown for distances. We now modify the Edge subdivision operation, which is used when dealing with graph embeddings (see Section \ref{sec: graphEmbeddings}).
\begin{enumerate}
	\item \emph{Edge subdivision}: Let $(u,v)$ be an edge of length $\ell(u,v)$. Delete $(u,v)$, introduce a new vertex $w$ and add edges $(u,w)$ and $(w,v)$, each of length $\ell(u,v)/2$. 
\end{enumerate}

We now prove an analogue to Lemma \ref{lemm: cutembedding}.

\begin{lemma} \label{lemm: DistancePreservation}
	Let $\rho$ be a node embedding and let $G_1$ and $\rho(G_1)$ be $k$-terminal graphs as defined in Section~\ref{sec: graphEmbeddings}. Then $\rho(G_1)$ preserves exactly all shortest path distances between terminal pairs.
\end{lemma}
\begin{proof}
	We can view each path obtained from the embedding as taking the edge corresponding to that path endpoints in $G_1$ and performing edge subdivisions finitely many times. We claim that such subdivisions preserve all terminal shortest paths.
	
	Indeed, let us consider a single edge subdivison for $(u,v)$ (the general claim then follows by induction on the number of edge subdivions). Fix $x,x' \in K$ and consider some shortest path $p(x,x')$ in $G_1$ that uses $(u,v)$. We can construct in $\rho(G_1)$ a path $q(x,x')$ of the same length as follows: traverse the subpath $p(x,u)$, traverse the edges $(u,w)$ and $(w,v)$ and finally traverse the subpath $p(v,x')$. It follows that $\sum_{e \in p(x,x')} \ell(e) = \sum_{e \in q(x,x')} \ell(e)$, and thus $d_{\rho(G_1)}(s,t) \leq d_{G_1}(s,t)$.
	
	On the other hand, fix $x,x' \in K$ and consider some shortest path $p'(x,x')$ in $\rho(G_1)$ that uses the two subdivided edges $(u,w)$ and $(w,v)$ (note that it cannot use only one of them). We can construct in $G_1$ a path $q'(x,x')$ of the same length as follows: traverse the subpath $p'(x,u)$, traverse the edge $(u,v)$ and finally traverse the subpath $p'(v,x')$. It follows that $\sum_{e \in p'(x,x')} \ell(e) = \sum_{e \in q'(x,x')} \ell(e)$ and thus $d_{G_1}(s,t) \leq d_{\rho(G_1)}(s,t)$. Combining the above gives the lemma.
\end{proof}

We next consider vertex splitting for graphs whose maximum degree is larger than three. For each vertex $v$ of degree $d > 3$ with $u_1,\ldots,u_d$ adjacent to $v$, we delete $v$ and introduce new vertices $v_1, \ldots, v_d$ along with edges $\{(v_i,v_{i+1}) : i = 1,\ldots,d-1\}$, each of length $0$. Furthermore, we replace the edges $\{(u_i,v) : i=1,\ldots,d\}$ with $\{(u_i,v_i) : i = 1,\ldots, d\}$, each of corresponding length. If $v$ is a terminal vertex, we make one of the $v_i$'s be a terminal vertex. An analogue to Claim \ref{claim: embedding} gives that the resulting graph preserves all terminal shortest path distances.

We finally note that whenever we add dummy edges of capacity $0$ in the cut setting, we replace them by edges of length $D+1$ in the distance setting, where $D$ is the sum over all edge lengths in the graph we consider. Since any shortest path in the graph does not use the added edges, the terminal shortest path remain unaffected. The above discussion leads to the following theorem. 

\begin{theorem} \label{thm: DistanceSparsifier}
	Let $G$ be a $k$-terminal planar graph where all terminals lie on the outer face. Then $G$ admits a quality-$1$ distance sparsifier of size $O(k^{2})$, which is also a planar graph. 
\end{theorem}
\subsection{Incompressibility of Distances in $k$-Terminal Graphs} \label{sec: lowerBound}

In this section we prove the following incompressibility result (i.e., Theorem~\ref{thm:incompressibility}) concerning the trade-off between quality and size of any compression function when estimating terminal distances in $k$-terminal graphs: for every $\varepsilon > 0$ and $t \geq 2$, there exists a (sparse) $k$-terminal $n$-vertex graph such that $k=o(n)$, and that any compression algorithm that approximates pairwise terminal distances within a factor of $t - \varepsilon$ or an additive error $2t-3$ must use $\Omega(k^{1+1/(t-1)})$ bits. Our lower bound is inspired by the work of Matou{\v{s}}ek~\cite{matousek96}, which has also been utilized in the context of distance oracles~\cite{ThorupZ05}. Our arguments rely on the recent extremal combinatorics construction (see~\cite{cheung2016}) that was used to prove lower bounds on the size of distance approximating minors. 

\paragraph*{Discussion on our result.} Note that for any $k$-terminal graph $G$, if we do not have any restriction on the structure of the distance sparsifier, then $G$ always admits a trivial quality $1$ distance sparsifier $H$ which is the complete weighted graph on $k$ terminals with each edge weight being equal to the distance between the two endpoints in $G$. Furthermore, by the well-known result of Awerbuch~\cite{awerbuch1985}, such a graph $H$ in turn admits a multiplicative $(2t-1)$-\emph{spanner} $H'$ with $O(k^{1+1/t})$ edges, that is, all the distances in $H$ are preserved up to a multiplicative factor of $2t-1$ in $H'$, for any $t\geq 1$. This directly implies that the $k$-terminal graph $G$ has a quality $2t-1$ distance sparsifier with $k$ vertices and $O(k^{1+1/t})$ edges. On the other hand, though \emph{unconditional} lower bounds of type similar to our result have been known for the number of edges of spanners~\cite{lazebnik1995,lazebnik1996,woodruff2006}, we are not aware of such lower bounds for the size of \emph{data structure} that preserves pairwise terminal distances for any $k$-terminal $n$-vertex graph when $k=o(n)$. In the extreme case when $k=n$ (i.e., all the vertices are terminals), the recent work by Abboud and Bodwin~\cite{abboud} shows that any data structure that preserves the distances with an additive error $t$ needs $\Omega(n^{4/3-\varepsilon})$ bits, for any $\varepsilon>0, t=O(n^{\delta})$ and $\delta=\delta(\varepsilon)$ (see also the follow-up work~\cite{abboud2017}).

We start by reviewing a classical notion in combinatorial design.
\begin{definition} [Steiner Triple System] Given a ground set $T=[k]$, an $(3,2)$-Steiner system (abbr. $(3,2)$-\emph{SS}) of $T$ is a collection of $3$-subsets of $T$, denoted by $\mathcal{S} = \{S_1,\ldots,S_r\}$, where $r = \binom{k}{2}\left/3\right.$, such that every $2$-subset of $T$ is contained in \emph{exactly} one of the $3$-subsets.
\end{definition}

\begin{lemma}[\cite{Wilson1975}] For infinity many $k$, the set $T=[k]$ admits an $(3,2)$-\emph{SS}.
\end{lemma}

Roughly speaking, our proof proceeds by forming a $k$-terminal bipartite graph, where terminals lie on one side and non-terminals on the other. The set of non-terminals will correspond to some subset of a Steiner Triple System $\mathcal{S}$, which will satisfy some \emph{certain} property. One can equivalently view such a graph as taking union over \emph{star} graphs. Before delving into details, we need to review a couple of other useful definitions and the construction from~\cite{cheung2016}. 

\paragraph*{\textbf{Detour Graph and Cycle.}} 
Let $k$ be an integer such that $T=[k]$ admits an $(3,2)$-SS. Let $\mathcal{S}$ be such an $(3,2)$-SS. We associate $\mathcal{S}=\{S_1,\ldots,S_r\}$ with a graph whose vertex set is $\mathcal{S}$. We refer to such graph as a \emph{detouring graph}.  By the definition of Steiner system, it follows that $|S_i \cap S_j|$ is either zero or one. Thus, two vertices $S_i$ and $S_j$ are adjacent in the detouring graph iff $|S_i \cap S_j|=1$. It is also useful to label each edge $(S_i, S_j)$ with the terminal in $S_i \cap S_j$. A \emph{detouring cycle} is a cycle in the detouring graph such that no two neighbouring edges in the cycles have the same terminal label. Observe that the detouring graph has other cycles which are not detouring cycles. 

Ideally, we would like to construct detouring graphs with long detouring cycles while keeping the size of the graph as large as possible. One trade-off is given in the following lemma.

\begin{lemma}[\cite{cheung2016}] \label{lemm: detouringGraph} For any integer $t \geq 3$, given a detouring graph with vertex set $\mathcal{S}$, there exists a subset $\mathcal{S}' \subset \mathcal{S}$ of cardinality $\Omega(k^{1+1/(t-1)})$ such that the induced graph on $\mathcal{S}'$ has no detouring cycles of size $t$ or less.
\end{lemma}

Now we are ready to prove our incompressibility result regarding approximately preserving terminal pairwise distances.
\paragraph*{\textbf{Proof of Theorem~\ref{thm:incompressibility}:}} 
Let $k$ be an integer such that $T=[k]$ admits an $(3,2)$-SS $\mathcal{S}$. Fix some integer $t \geq 3$, some positive constant $c$ and use Lemma \ref{lemm: detouringGraph} to construct a subset $\mathcal{S}'$ of $\mathcal{S}$ of size $\Omega(k^{1+1/(t-1)})$ such that the induced graph on $\mathcal{S}'$ has no detouring cycles of size $t$ or less. We may assume w.l.o.g. that $\ell = |\mathcal{S}'| = c \cdot k^{1+1/(t-1)}$ (this can be achieved by repeatedly removing elements from $\mathcal{S}'$, as the property concerning the detouring cycles is not destroyed). Fix some ordering among $3$-subsets of $\mathcal{S}'$ and among terminals in each $3$-subset.

We define the $k$-terminal graph $G$ as follows:
\begin{itemize}
\item For each $e_i \in \mathcal{S}'$ create a non-terminal vertex $v_i$. Let $V_{\mathcal{S}'}$ denote the set of such vertices. The vertex set of $G$ is $T \cup V_{\mathcal{S}'}$, where $T=[k]$ denotes the set of terminals.
\item For each $e_i \in \mathcal{S}'$, connect $v_i$ to the three terminals $\{x^{i}_{1},x^{i}_{2},x^{i}_{3}\}$ belonging to $e_i$, i.e., add edges $(v_i,x^{i}_j)$, $j=1,2,3$.
\end{itemize}
Note that $G$ is sparse since both the number of vertices and edges are $\Theta(\ell)$, and it also holds that $k=o(|V(G)|)$.

For any subset $R \subseteq \mathcal{S}'$, we define the subgraph $G_R=(V(G), E_R)$ of $G$ as follows. For each $e_i \in S'$, if $e_i \in R$, perform no changes. If $e_i \not \in R$, delete the edge $(v_i,x^{i}_1)$. Note that there are $2^{\ell}$ subgraphs $G_R$. We let $\mathcal{G}$ denote the family of all such subgraphs. 

We say a terminal pair $(x,x')$ \emph{respects $\mathcal{S'}$} if in the $(3,2)$-SS $\mathcal{S}$, the unique $3$-subset $e$ that contains $x$ and $x'$ belongs to $\mathcal{S'}$. Given $R \subseteq \mathcal{S}'$ and some terminal pair $(x,x')$, we say that $R$ \emph{covers} $(x,x')$ if both $x$ and $x'$ are connected to some non-terminal $v$ in $G_R$. 

\begin{claim}~\label{claim:cover} 
	For all $R \subseteq \mathcal{S}'$ and terminal pairs $(x,x')$ covered by $R$ we have that $d_{G_R}(x,x')=2$.
\end{claim}
\begin{proof}
By the definition of Steiner system and the construction of $G_R$, the shortest path between $x$ and $x'$ is simply a $2$-hop path, i.e., $d_{G_R}(x,x') = 2$. 
\end{proof}

\begin{claim}~\label{claim:non_cover}
For all $R \subseteq \mathcal{S}'$ and any terminal pair $(x,x')$ that respects $\mathcal{S}'$ and is \emph{not} covered by $R$, we have that $d_{G_R}(x,x')\geq 2t$.
\end{claim}
\begin{proof}
Since $(x,x')$ respects $\mathcal{S'}$, there exists $e_i=(x^i_1,x^i_2,x^i_3)\in \mathcal{S}'$ that contains both $x$ and $x'$. By construction of $G_R$ and the fact that $(x,x')$ is not covered by $R$, it follows that $e_i\in \mathcal{S}'\setminus R$, and one of $x,x'$ corresponds to $x^i_1$ and the other corresponds to $x^i_2$ or $x^i_3$. W.l.o.g., we assume $x=x^i_1$ and $x'=x^i_2$. Note that there is no edge connecting $x^{i}_1$ with the non-terminal $v_i$ that corresponds to $e_i$. Now by Lemma~\ref{lemm: detouringGraph}, the detouring graph induced on $\mathcal{S}'$ has no detouring cycles of size $t$ or less, which implies that any other simple path between $x^{i}_1$ and $x^{i}_2$ in $G$ must pass through at least $t-1$ other terminals. Let $w_1,\ldots,w_{t-1}$ be such terminals and let $P:=x^{i}_1 \rightarrow w_1,\ldots,w_{t-1} \rightarrow x^{i}_2$ denote the corresponding path, ignoring the non-terminals along the path. Between any consecutive terminal pairs in $P$, the shortest path is at least $2$. Thus, the length of $P$ is at least $2t$, i.e., $d_{G_R}(x^{i}_1,x^{i}_2) \geq 2t$.
\end{proof}

Fix any two subsets $R_1, R_2 \subseteq \mathcal{S}'$ with $R_1 \neq R_2$. It follows that there exists a $3$-subset $e_i= (x^i_1,x^i_2,x^i_3)\in \mathcal{S}'$ such that either $e \in R_1 \setminus R_2$ or $e \in R_2 \setminus R_1$. Assume w.l.o.g that $e \in R_2 \setminus R_1$. Note that $(x^i_1,x^i_2)$ respects $\mathcal{S}'$ and it is covered in $R_2$ but not in $R_1$. By~Claim~\ref{claim:cover} and~\ref{claim:non_cover}, it holds that $d_{G_{R_{2}}}(x^i_1,x^i_2) = 2$ and $d_{G_{R_1}}(x^i_1,x^i_2) \geq 2t$. In other words, there exists a set $\mathcal{G}$ of $2^{\ell}$ different subgraphs on the same set of nodes $V(G)$ satisfying the following property: for any $G_1,G_2 \in \mathcal{G}$, there exists a terminal pair $(x,x')$ such that the distances between $x$ and $x'$ in $G_1$ and $G_2$ differ by at least a $t$ factor as well as by at least $2t-2$. On the other hand, for any compression function that approximates terminal path distances within a factor of $t-\varepsilon$ or an additive error $2t-3$ and produces a bitstring with less than $\ell$ bits, there exist two different graphs $G_1,G_2 \in \mathcal{G}$ that map to the same bit string. Hence, any such compression function must use at least $\Omega(\ell) = \Omega(k^{1+1/(t-1)})$ bits if we want to preserve terminal distances within a $t-\varepsilon$ factor or an additive error $2t-3$. 



To complete our argument, we need to show the claim for quality $t=2$. The only significant modification we need is the usage of an $(3,2)$-SS in the construction of graph $G$ (instead of using a subset of it). The remaining details are similar to the above proof and we omit them here.







\bibliographystyle{alpha}
\bibliography{literature}

\newcommand{\etalchar}[1]{$^{#1}$}
\begin{thebibliography}{AGMW18}

\bibitem[AB16]{abboud}
Amir Abboud and Greg Bodwin.
\newblock The 4/3 additive spanner exponent is tight.
\newblock In {\em Proc. of the 48th STOC}, pages 351--361, 2016.

\bibitem[AB18]{AB18reachability}
Amir Abboud and Greg Bodwin.
\newblock Reachability preservers: New extremal bounds and approximation
  algorithms.
\newblock In {\em Proc. of the 29th SODA}, 2018.
\newblock available at Arxiv: CoRR abs/1710.11250.

\bibitem[ABP17]{abboud2017}
Amir Abboud, Greg Bodwin, and Seth Pettie.
\newblock A hierarchy of lower bounds for sublinear additive spanners.
\newblock In {\em Proc. of the 28th SODA}, pages 568--576, 2017.

\bibitem[AGK14]{andoni}
Alexandr Andoni, Anupam Gupta, and Robert Krauthgamer.
\newblock Towards {(1+ $\varepsilon$)}-approximate flow sparsifiers.
\newblock In {\em Proc. of the 25th SODA}, pages 279--293, 2014.

\bibitem[AGMW18]{AbboudGMW17}
Amir Abboud, Pawel Gawrychowski, Shay Mozes, and Oren Weimann.
\newblock Near-optimal compression for the planar graph metric.
\newblock In {\em Proc. of the 29th SODA}, 2018.
\newblock available at arXiv: CoRR abs/1703.04814.

\bibitem[AGU72]{AhoGU72}
Alfred~V. Aho, M.~R. Garey, and Jeffrey~D. Ullman.
\newblock The transitive reduction of a directed graph.
\newblock {\em {SIAM} J. Comput.}, 1(2):131--137, 1972.

\bibitem[Awe85]{awerbuch1985}
Baruch Awerbuch.
\newblock Complexity of network synchronization.
\newblock {\em J. ACM}, 32(4):804--823, 1985.

\bibitem[BK96]{BenczurK96}
Andr{\'{a}}s~A. Bencz{\'{u}}r and David~R. Karger.
\newblock Approximating \emph{s-t} minimum cuts in
  \emph{{\~{O}}}(\emph{n}\({}^{\mbox{2}}\)) time.
\newblock In {\em Proc. of the 28th STOC}, pages 47--55, 1996.

\bibitem[Bod17]{Bodwin17}
Greg Bodwin.
\newblock Linear size distance preservers.
\newblock In {\em Proc. of the 28th SODA}, pages 600--615, 2017.

\bibitem[CE06]{coppersmithE06}
Don Coppersmith and Michael Elkin.
\newblock Sparse sourcewise and pairwise distance preservers.
\newblock {\em {SIAM} J. Discrete Math.}, 20(2):463--501, 2006.

\bibitem[CGH16]{cheung2016}
Yun~Kuen Cheung, Gramoz Goranci, and Monika Henzinger.
\newblock Graph minors for preserving terminal distances approximately - {Lower
  and Upper Bounds}.
\newblock In {\em Proc. of the 43rd ICALP}, pages 131:1--131:14, 2016.

\bibitem[CGN{\etalchar{+}}06]{chekuri2006embedding}
Chandra Chekuri, Anupam Gupta, Ilan Newman, Yuri Rabinovich, and Alistair
  Sinclair.
\newblock Embedding k-outerplanar graphs into l1.
\newblock {\em SIAM J. Discrete Math.}, 20(1):119--136, 2006.

\bibitem[Che18]{Cheung17}
Yun~Kuen Cheung.
\newblock Steiner point removal - distant terminals don't (really) bother.
\newblock In {\em Proc. of the 29th SODA}, 2018.
\newblock available at Arxiv: CoRR abs/1703.08790.

\bibitem[Chu12a]{juliasteiner}
Julia Chuzhoy.
\newblock On vertex sparsifiers with steiner nodes.
\newblock In {\em Proc. of the 44th STOC}, pages 673--688, 2012.

\bibitem[Chu12b]{chuzhoy2012routing}
Julia Chuzhoy.
\newblock Routing in undirected graphs with constant congestion.
\newblock In {\em Proc. of the 44th STOC}, pages 855--874, 2012.

\bibitem[CIM98]{curtis98}
Edward~B Curtis, David Ingerman, and James~A Morrow.
\newblock Circular planar graphs and resistor networks.
\newblock {\em Linear algebra and its applications}, 283(1):115--150, 1998.

\bibitem[CKS09]{chekuri09}
Chandra Chekuri, Sanjeev Khanna, and F~Bruce Shepherd.
\newblock Edge-disjoint paths in planar graphs with constant congestion.
\newblock {\em SIAM J. Comput.}, 39(1):281--301, 2009.

\bibitem[CLLM10]{charikar}
Moses Charikar, Tom Leighton, Shi Li, and Ankur Moitra.
\newblock Vertex sparsifiers and abstract rounding algorithms.
\newblock In {\em Proc. of the 51th FOCS}, pages 265--274, 2010.

\bibitem[CSW10]{chekuri2010flow}
Chandra Chekuri, F~Bruce Shepherd, and Christophe Weibel.
\newblock Flow-cut gaps for integer and fractional multiflows.
\newblock In {\em Proc. of the 21st SODA}, pages 1198--1208, 2010.

\bibitem[CXKR06]{chan06}
T-H~Hubert Chan, Donglin Xia, Goran Konjevod, and Andrea Richa.
\newblock A tight lower bound for the steiner point removal problem on trees.
\newblock In {\em Proc. of the 9th APPROX/RANDOM}, pages 70--81, 2006.

\bibitem[DDK{\etalchar{+}}17]{DaubelDKMS17}
Karl D{\"{a}}ubel, Yann Disser, Max Klimm, Torsten M{\"{u}}tze, and Frieder
  Smolny.
\newblock Distance-preserving graph contractions.
\newblock {\em CoRR}, abs/1705.04544, 2017.

\bibitem[DS07]{DiksS07}
Krzysztof Diks and Piotr Sankowski.
\newblock Dynamic plane transitive closure.
\newblock In {\em Proc. of the 15th ESA}, pages 594--604, 2007.

\bibitem[EGK{\etalchar{+}}14]{englert10}
Matthias Englert, Anupam Gupta, Robert Krauthgamer, Harald R{\"{a}}cke, Inbal
  Talgam{-}Cohen, and Kunal Talwar.
\newblock Vertex sparsifiers: New results from old techniques.
\newblock {\em SIAM J. Comput.}, 43(4):1239--1262, 2014.

\bibitem[Fil18]{Filtser17}
Arnold Filtser.
\newblock Steiner point removal with distortion o(log k).
\newblock In {\em Proc. of the 29th SODA}, 2018.
\newblock available at Arxiv: CoRR abs/1706.08115.

\bibitem[FP93]{feo1993}
Thomas~A Feo and J~Scott Provan.
\newblock Delta-wye transformations and the efficient reduction of two-terminal
  planar graphs.
\newblock {\em Operations Research}, 41(3):572--582, 1993.

\bibitem[Git91]{Gitler91}
Isidoro Gitler.
\newblock {\em Delta-Wye-Delta Transformations: Algorithms and Applications}.
\newblock PhD thesis, Department of Combinatorics and Optimization, University
  of Waterloo, 1991.

\bibitem[GR16]{GoranciR16}
Gramoz Goranci and Harald R{\"{a}}cke.
\newblock Vertex sparsification in trees.
\newblock In {\em Proc. of the 14th WAOA}, pages 103--115, 2016.

\bibitem[GR17]{GajjarR17}
Kshitij Gajjar and Jaikumar Radhakrishnan.
\newblock Distance-preserving subgraphs of interval graphs.
\newblock In {\em Proc. of the 25th ESA}, pages 39:1--39:13, 2017.

\bibitem[Gup01]{gupta01}
Anupam Gupta.
\newblock Steiner points in tree metrics don't (really) help.
\newblock In {\em Proc. of the 12th SODA}, pages 220--227, 2001.

\bibitem[HKNR98]{HagerupKNR98}
Torben Hagerup, Jyrki Katajainen, Naomi Nishimura, and Prabhakar Ragde.
\newblock Characterizing multiterminal flow networks and computing flows in
  networks of small treewidth.
\newblock {\em J. Comput. Syst. Sci.}, 57(3):366--375, 1998.

\bibitem[KKN15]{kamma2015}
Lior Kamma, Robert Krauthgamer, and Huy~L Nguyen.
\newblock Cutting corners cheaply, or how to remove steiner points.
\newblock {\em SIAM J. Comput.}, 44(4):975--995, 2015.

\bibitem[KKS05]{katriel2005reachability}
Irit Katriel, Martin Kutz, and Martin Skutella.
\newblock Reachability substitutes for planar digraphs.
\newblock In {\em Technical Report MPI-I-2005-1-002}. Max-Planck-Institut
  f{\"u}r Informatik, 2005.

\bibitem[KNZ14]{distancepreserving}
Robert Krauthgamer, Huy~L Nguyen, and Tamar Zondiner.
\newblock Preserving terminal distances using minors.
\newblock {\em SIAM J. Discrete Math.}, 28(1):127--141, 2014.

\bibitem[KPZP17]{karpov2017exponential}
Nikolay Karpov, Marcin Pilipczuk, and Anna Zych-Pawlewicz.
\newblock An exponential lower bound for cut sparsifiers in planar graphs.
\newblock {\em The 12th International Symposium on Parameterized and Exact
  Computation (IPEC)}, 2017.

\bibitem[KR13]{KrauthgamerR13}
Robert Krauthgamer and Inbal Rika.
\newblock Mimicking networks and succinct representations of terminal cuts.
\newblock In {\em Proc. of the 24th SODA}, pages 1789--1799, 2013.

\bibitem[KR14]{KhanR14}
Arindam Khan and Prasad Raghavendra.
\newblock On mimicking networks representing minimum terminal cuts.
\newblock {\em Inf. Process. Lett.}, 114(7):365--371, 2014.

\bibitem[KR17]{krauthgamer2017refined}
Robert Krauthgamer and Inbal Rika.
\newblock Refined vertex sparsifiers of planar graphs.
\newblock {\em CoRR}, abs/1702.05951, 2017.

\bibitem[KZ12]{KrauthgamerZ12}
Robert Krauthgamer and Tamar Zondiner.
\newblock Preserving terminal distances using minors.
\newblock In {\em Proc. of the 39th ICALP}, pages 594--605, 2012.

\bibitem[LM10]{leighton}
Frank~Thomson Leighton and Ankur Moitra.
\newblock Extensions and limits to vertex sparsification.
\newblock In {\em Proc. of the 42nd STOC}, pages 47--56, 2010.

\bibitem[LMM13]{LeeMM13}
James~R Lee, Manor Mendel, and Mohammad Moharrami.
\newblock A node-capacitated okamura-seymour theorem.
\newblock In {\em Proc. of the 45th STOC}, pages 495--504, 2013.

\bibitem[LUW95]{lazebnik1995}
Felix Lazebnik, Vasiliy~A Ustimenko, and Andrew~J Woldar.
\newblock A new series of dense graphs of high girth.
\newblock {\em Bulletin of the American mathematical society}, 32(1):73--79,
  1995.

\bibitem[LUW96]{lazebnik1996}
Felix Lazebnik, Vasiliy~A Ustimenko, and Andrew~J Woldar.
\newblock A characterization of the components of the graphs d (k, q).
\newblock {\em Discrete Mathematics}, 157(1-3):271--283, 1996.

\bibitem[Mat96]{matousek96}
Ji{\v{r}}{\'\i} Matou{\v{s}}ek.
\newblock On the distortion required for embedding finite metric spaces into
  normed spaces.
\newblock {\em Israel Journal of Mathematics}, 93(1):333--344, 1996.

\bibitem[MM10]{mm10}
Konstantin Makarychev and Yury Makarychev.
\newblock Metric extension operators, vertex sparsifiers and lipschitz
  extendability.
\newblock In {\em Proc. of the 51th FOCS}, pages 255--264, 2010.

\bibitem[Moi09]{moitra09}
Ankur Moitra.
\newblock Approximation algorithms for multicommodity-type problems with
  guarantees independent of the graph size.
\newblock In {\em Proc. of the 50th FOCS}, 2009.

\bibitem[OS81]{OkamuraS81}
Haruko Okamura and Paul~D. Seymour.
\newblock Multicommodity flows in planar graphs.
\newblock {\em J. Comb. Theory, Ser. {B}}, 31(1):75 -- 81, 1981.

\bibitem[ST11]{SpielmanT04}
Daniel~A. Spielman and Shang{-}Hua Teng.
\newblock Spectral sparsification of graphs.
\newblock {\em {SIAM} J. Comput.}, 40(4):981--1025, 2011.

\bibitem[Sub93]{Subramanian93}
Sairam Subramanian.
\newblock A fully dynamic data structure for reachability in planar digraphs.
\newblock In {\em Proc. of the 1st ESA}, pages 372--383, 1993.

\bibitem[Tar72]{Tarjan72}
Robert~Endre Tarjan.
\newblock Depth-first search and linear graph algorithms.
\newblock {\em {SIAM} J. Comput.}, 1(2):146--160, 1972.

\bibitem[Tho04]{Thorup04}
Mikkel Thorup.
\newblock Compact oracles for reachability and approximate distances in planar
  digraphs.
\newblock {\em J. {ACM}}, 51(6):993--1024, 2004.

\bibitem[TT89]{TamassiaT89}
Roberto Tamassia and Ioannis~G Tollis.
\newblock Planar grid embedding in linear time.
\newblock {\em IEEE Trans. Circuits Syst.}, 36(9):1230--1234, 1989.

\bibitem[TZ05]{ThorupZ05}
Mikkel Thorup and Uri Zwick.
\newblock Approximate distance oracles.
\newblock {\em J. ACM}, 52(1):1--24, 2005.

\bibitem[Val81]{Valiant81}
Leslie~G. Valiant.
\newblock Universality considerations in {VLSI} circuits.
\newblock {\em {IEEE} Trans. Computers}, 30(2):135--140, 1981.

\bibitem[Wil75]{Wilson1975}
Richard~M. Wilson.
\newblock An existence theory for pairwise balanced designs, {III:} proof of
  the existence conjectures.
\newblock {\em J. Comb. Theory, Ser. {A}}, 18(1):71--79, 1975.

\bibitem[Woo06]{woodruff2006}
David~P Woodruff.
\newblock Lower bounds for additive spanners, emulators, and more.
\newblock In {\em Proc. of the 47th FOCS}, pages 389--398, 2006.

\end{thebibliography}



\end{document}